\documentclass{article}

\usepackage{PRIMEarxiv}

\usepackage[utf8]{inputenc} 
\usepackage[T1]{fontenc}    
\usepackage{hyperref}       
\usepackage{url}            
\usepackage{booktabs}       
\usepackage{amsfonts}       
\usepackage{nicefrac}       
\usepackage{microtype}      
\usepackage{fancyhdr}       
\usepackage{graphicx}       
\graphicspath{{media/}}     

\usepackage{subcaption}

\usepackage{amsmath,amssymb, amsthm}
\usepackage{lipsum}
\usepackage{amsfonts}
\usepackage{graphicx}
\usepackage{epstopdf}
\usepackage{algorithmic}
\newtheorem{definition}{Definition}
\newtheorem{theorem}{Theorem}
\newtheorem{lemma}[theorem]{Lemma}
\newtheorem*{remark}{Remark}
\newtheorem{corollary}{Corollary}[theorem]
\newtheorem{proposition}{Proposition}
\usepackage{algorithm}
\usepackage{placeins}

\title{A Physics Based Surrogate Model in Bayesian Uncertainty Quantification involving Elliptic PDEs
}

\author{
  A. Galaviz, J. A. Christen \\
  Centro de Investigación en Matemáticas (CIMAT) \\
  Guanajuato, México\\
  \texttt{\{anel.galaviz, jac\}@cimat.mx} \\
   \And
  A. Capella \\
  Instituto de Matemáticas \\
  Universidad Nacional Autónoma de México (UNAM) \\
  Ciudad de México, México\\
  \texttt{capella@im.unam.mx} \\
}

\begin{document}
\maketitle

\begin{abstract}
The paper addresses Bayesian inferences in inverse problems with uncertainty quantification involving a computationally expensive forward map associated with solving a partial differential equations. To mitigate the computational cost, the paper proposes a new surrogate model informed by the physics of the problem, specifically when the forward map involves solving a linear elliptic partial differential equation. The study establishes the consistency of the posterior distribution for this surrogate model and demonstrates its effectiveness through numerical examples with synthetic data. The results indicate a substantial improvement in computational speed, reducing the processing time from several months with the exact forward map to a few minutes, while maintaining negligible loss of accuracy in the posterior distribution.
\end{abstract}


\section{Introduction}
The goal of Bayesian inference with uncertainty quantification applied to inverse problems is to calculate a posterior distribution for the unknown parameters. In instances involving complex models where obtaining an exact solution for the posterior distribution is impractical, approximations are often employed, typically generated by sampling algorithms like Markov Chain Monte Carlo (MCMC) methods \cite{Robert, Liu}. Each step of the MCMC algorithm involves evaluating the posterior distribution, which includes likelihood and forward map evaluations. In cases where forward map evaluations are computationally demanding, a common strategy is to utilize a surrogate model for the forward map. While this approach reduces computational time and facilitates inference, it introduces numerical and approximation errors.

In science, multiple computational models, each with varying costs and fidelities, are employed to describe a system of interest. High-fidelity models offer accurate descriptions but come with significant computational cost. For instance, models involving complex Partial Differential Equations (PDEs) with multiple parameters and fine meshes require considerable computational time, ranging from minutes to days, for a single forward simulation.  Applications, such as design optimization or Bayesian inversion with uncertainty quantification, demand thousands to millions of model evaluations  \cite{Stuart-2010}. Unfortunately, the associated computational costs often render these tasks practically inaccessible. These challenges motivate the need for the development of computationally efficient surrogate models. Typically, surrogate models are computationally fast but less accurate, so we describe them as low-fidelity. Choosing a low-fidelity model depends heavily on the particular application. A common practice involves leveraging both low and high-fidelity models to expedite simulations while gauging accuracy, necessitating a well-defined model management strategy \cite{Peherstorfer-2018}.

Surrogate models can be broadly categorized into four main types: projection-based reduced models, data-fit interpolation and regression models, machine-learning-based models, and simplified surrogate models \cite{Peherstorfer-2018, Asher-2015}. The Projection-based reduced surrogate model aims to reduce the complexity of the Forward Map by altering the mathematical structure of the forward's map numerical scheme\cite{LAWRENCE, Rozza, Bui-Thanh-2008}. The basic idea is to project the solutions of the elliptic partial differential equations (PDEs) into a low-dimensional subspace to derive a reduced model. The success of this projection step relies on a good understanding of the Forward Map's structure. Various methods exist for constructing the low-dimensional subspace \cite{Benner}. One commonly employed technique is proper orthogonal decomposition (POD) \cite{Sirovich, Berkooz}. POD methods utilize``snapshots," which are state vectors of the Forward Map (from the high-fidelity model) corresponding to selected inputs, to establish a basis for the low-dimensional subspace. POD is favored for its versatility, as it is applicable to a wide range of problems, including those that are time-dependent and nonlinear \cite{Chaturantabut}. Depending on the nature of the problem at hand, these reduced basis models may incorporate inexpensive a posteriori error estimators for the outputs of the reduced model \cite{HUYNH2007473}.

Data-fit surrogate models are directly derived from the inputs and outputs of the high-fidelity model. These models, obtained through interpolation, regression, or black-box techniques, represent approximations to the Forward Map evaluations as linear combinations of basis functions. The construction involves fitting coefficients through interpolation or regression to the corresponding inputs and outputs of the high-fidelity model. The choice of bases for interpolation or regression significantly influences the quality of the surrogate model's approximation to the Forward Map. Radial basis functions, commonly used for modeling data, are one such type of basis function \cite{Rasmussen, FORRESTER200950}. This category includes deep neural network models capable of addressing high dimensionality by approximating the solutions to specific classes of partial differential equations \cite{Berner-2020, grohs2018proof, Jentzen, KOURAKOS-2009}.

Machine-learning-based surrogate models serve as efficient substitutes for the Forward Map. These models aim to capture intricate relationships between input and output variables specific to a Forward Map. Various machine learning-based models, including support vector machines (SVMs), neural networks, and random forests, are known for their flexibility and ability to adapt to complex data patterns \cite{He-2019, Cai-2021, Cortes-1995, RAISSI2019686, Deveney-2021}.
One notable advantage is the continuous adaptability of these models when acquiring new data, allowing them to evolve as more information about the system becomes available. However, it is important to note that many machine-learning models involve hyperparameters that must be tuned to optimize the surrogate model's performance for a specific task.

Simplified models are derived by leveraging knowledge of the Forward Map model to provide a more agile and efficient representation of the system without substantial loss of accuracy compared to the Forward Map. Various techniques can be employed, including representing the system in a lower-dimensional space, fitting simple analytical functions to observed data, using a linearized version of a nonlinear PDE in a specific regime, or employing simpler mathematical formulations that capture the essence of the system's behavior \cite{Ashby-1996, SAIED-1998, SHI-2012, HOU-1997}.

All the above surrogate model construction strategies present both advantages and drawbacks. The Projection-based reduced surrogate model faces limitations tied to the critical choice of a function basis, impacting accuracy and introducing sensitivity to the training parameter space, particularly with complex function bases \cite{Santner-2003, Buckley-1993}. Similar to Projection-based models, data-fit surrogate models simplify reality, risking the loss of fine details in the Forward Map. Their quality heavily depends on the accuracy of high-fidelity data\cite{Iooss-2015, Wagner-1995}, posing challenges with unrepresentative or erroneous data. Machine-learning-based surrogate models offer advantages but are susceptible to overfitting, lack interpretability, and may lack direct uncertainty estimates crucial in Uncertainty Quantification. Simplified models struggle to represent complex phenomena accurately, limiting dynamic systems' utility. Validating simplified models is challenging, as their representation of complex phenomena may be underestimated.

In this paper, we introduce an approach aimed at constructing low-computational-cost surrogate models to facilitate Bayesian inference on parameters while maintaining control over the error in the posterior distribution.

The subsequent sections outline a new method termed \emph{Physics-Based Surrogate Models}. This surrogate model closely resembles a data-fit low-fidelity model approximation, yet it does not fit neatly into the classification above. Broadly speaking, we leverage the high-fidelity model's analytical properties and structure to determine our surrogate's coefficients.
Our approach is physics-based, as the model's structure plays a crucial role in shaping the surrogate model coefficients. Once the coefficients of the surrogate model are defined, only a relatively small number of high-fidelity simulations are required to establish the numerical approximation of the surrogate. An advantages of our method lies in the offline execution of high-fidelity simulations during a preprocessing step, eliminating the need for any further in-line high-fidelity model evaluations. Furthermore, the convergence and accuracy in our method result from the precise definition of coefficients and the number of high-fidelity offline simulations. The simulation selection becomes critical and corresponds to an experimental design strategy, which we elaborate on later in the text. The algorithm resulting from our approach involves evaluating a few algebraic equations for low-fidelity model assessments. This characteristic accelerates the outer-loop scheme by orders of magnitude, making it viable for addressing Bayesian uncertainty quantification problems. In loose terms, we propose a physics-informed, data-fit, low-fidelity surrogate model.

The paper is organized as follows. Section~\ref{sec:general} introduces the general setting and notation for the paper.   Section~\ref{sec:main} presents the proposed physics-informed surrogate model, 
Several results are presented, including guidelines for establishing the surrogate, characteristics of the experimental design employed for PDE evaluation, error estimates, and a proof of consistency for the posterior's approximation. Section~\ref{sec:examples} offers three examples of inverse problems in Electrical Impedance Tomography within a disk. Lastly, Section~\ref{sec:discussion} concludes the paper with a discussion and considerations for future work.


\section{General setting and notations}
\label{sec:general}
In this paper, we concern ourselves with regression models of the form 
\begin{equation}\label{modelo}
    y_i = \mathcal{H}_i (\mathcal{F}(\theta))+ \epsilon, 
    \qquad \epsilon_i \sim N(0, \sigma^2),
\end{equation} 
where $i=1,\ldots,m$ and the unknown parameter $\theta$ belongs to an admissible set of parameter $\Theta$.  In principle, the parameter set $\Theta$ belong to a infinite dimensional space, although we will only consider the finite-dimensional case.  Hence, we assume that $\Theta$ is an open set in $\mathbb{R}^d$.  The operator $\mathcal{F} : \Theta \to V$  is the Forward Map (FM), and the operator $\mathcal{H}: V \to \mathbb{R}^m$ is the observation operator. We assume a component-independent additive Gaussian observational error $\epsilon_i$ for definitiveness, but its particular form is not essential in our methods. 

The composition $\mathcal{H} \circ \mathcal{F}$ characterizes the regression problem, and in many practical cases, $\mathcal{F}$ involves solving a system of ordinary or partial differential equations. Examples are found  in ecology \cite{hutchinson2017}, image processing \cite{Giovannelli, Chama-2012}, and heat transfer \cite{Kaipio-2011, Berger-2016}, among others. 

In the present paper, we focus on the case where $\mathcal{F}(\theta)$ is the solution $u_\theta$ of a linear elliptic partial differential equation (PDE) on a domain $\Omega\subset \mathbb{R}^n$ with some boundary conditions in a Sobolev Space $V$. Namely, $\mathcal{F}(\theta) = u_\theta$ is the unique (theoretical) weak solution to the elliptic PDE equation that involves the parameter $\theta$ in some way. We will also say that $u_\theta$ is the variational solution to the PDE problem and can be a vector or a scalar function. In this setting, the Finite element method (FEM) is the standard approach to compute a numerical approximation to $u_\theta$. In addition, we assume that $\mathcal{H}_i$ is a linear operator evaluated on the PDE solution $u_\theta$. From now on, and without loss of generality, we assume that $\mathcal{H}_i$ is the evaluation operator at some set of given points $\{x_i \}_{i=1}^m\subset \Omega$, namely $u_\theta(x_i)$ for $i=1,\ldots,m$. Hence, $\mathcal{G}$ is the composition of  
$$ 
\theta \xrightarrow{\mathcal{F}} u_\theta \xrightarrow{\mathcal{H}} \{ u_\theta(x_1), u_\theta(x_2), \ldots, u_\theta(x_m) \},
$$
and, component-wise we write $\mathcal{G}_i(\theta) = \mathcal{H}_i(u_\theta)=u_\theta(x_i)$.


\section{Main results} \label{sec:main}

\subsection{Physics-Based Surrogate Model Method}
\label{sec:PBSMM}
The method for constructing Physics-Based Surrogate Models applies under the following three assumptions: 

First, we assume that $\mathcal{F}$ maps a parameter $\theta$ to the weak solution of an elliptic PDE of form 
\begin{equation}\label{classicalpde}
L_\theta u = f \quad\text{ in } \Omega,
\end{equation}
subject to some boundary conditions. The operator $L_\theta$ is uniformly elliptic and depends on $\theta$ in some yet unspecified form. In this case, the corresponding weak formulation of \eqref{classicalpde} on a Hilbert space $V$ is given by: finding $u \in V$ such that 
\begin{equation}\label{bilform}
 \mathcal{A}_\theta ( u, v) ={\cal L} (v) \text{ for every } v\in V.
\end{equation}
The bilinear form $\mathcal{A}_\theta(\cdot,\cdot)$ is also parameter-dependent, and we write $u_\theta$ for the weak solution of \eqref{bilform} for a given $\theta$. Moreover, due to the uniform ellipticity of $L_\theta$, the bilinear operator $\mathcal{A}_\theta$ satisfies the assumptions of the Lax-Milgram theorem \cite[p.~335]{salsa}, and the weak solution's existence and uniqueness are assured.

 Second, regarding the $\mathcal{A}_{\theta}$ dependence on $\theta$, we assume the following estimate 
\begin{equation} \label{modulodecontinuidad}
\sup_{v\in V, \Vert v \Vert_V \leq 1} \vert \mathcal{A}_{\theta}(u,v)-\mathcal{A}_{\theta'}(u,v) \vert \leq C G(\theta, \theta')  F(u), \end{equation}
where $C>0$ is a constant, and the function $G$ depends on $\theta$ and $\theta'$, but it is independent of $u$. Furthermore, we assume that $G$ is symmetric on its entries and $G(\theta, \theta') \to 0$ when $\theta \to \theta'$. Notice that the latter implies the continuity of $\mathcal{A}_{\theta}$ with respect to the parameters $\theta$. Finally, the function $F(u)$ is assumed to be bounded for any $u\in V$. Estimate \eqref{modulodecontinuidad} is a key structural hypothesis in our method. 

Third, there exists a design. That is a finite set of parameter values,  $\mathcal{D} = \{\theta_1,\ldots,\theta_n \}\subset\Theta$, and its corresponding set of weak PDE solutions to \eqref{classicalpde}, namely $\mathcal{U}= \{u_{\theta_1}, u_{\theta_2}, \ldots, u_{\theta_n} \}$, and we call its element surrogate solutions. 
\medskip

\begin{remark}
We implicitly assume that all three latter assumptions hold for the rest of the paper.
\end{remark}

\medskip
Next, we introduce the set $\mathcal{I}_\theta$ of $k$ ``nearest neighbors" of $\theta$ according to $G$.  The set $\mathcal{I}_\theta$ is a set of indices, and we proceed in two steps to compute it. First, we order the elements of the design $\mathcal{D}$ in terms of its $G(\theta,\cdot)$ value, namely, 
$$
G(\theta,\theta_{\tau(1)}) \leq G(\theta,\theta_{\tau(2)}) \leq \ldots \leq G(\theta,\theta_{\tau(n)}),
$$ 
where $\tau$ is some permutation of the indexes $\{1,2,\dots,n\}$. Second, we select the indices of the first $k\le n$ terms of the permutation $\tau$. Namely, $\mathcal{I}_\theta = \{i  \mid \tau(i)\le k \}$ where clearly  $\vert \mathcal{I}_\theta \vert = k$. Notice that we do not assume that $G$ is a distance, but we can use it to order the elements in $\mathcal{D}$ and compute $\mathcal{I}_\theta$.

\medskip
\begin{definition}\label{def:surrogate}
 We say that  $\hat{\mathcal{F}}:\Theta \to V$ is the \textbf{Physics-based surrogate model definition} corresponding to $\mathcal{F}$ provided:
\begin{enumerate} 
\item[(a)] If $\theta \in \mathcal{D}$ then 
\begin{equation}
\hat{\mathcal{F}}(\theta) = u_\theta.
\end{equation}

\item[(b)] If $\theta\notin \mathcal{D}$, the surrogate model is a convex combination of the surrogate solutions of the form
\begin{equation}\label{surrogatemodel} 
\hat{\mathcal{F}}(\theta) = \hat{u}_\theta = \sum_{i \in \mathcal{I}_\theta} \alpha_i(\theta)u_{\theta_i},
\end{equation}
where $\alpha_i(\theta)$ are the surrogate's coefficients that we discuss in the following section,  and $\mathcal{I}_\theta$ is the set of $k$ ``nearest neighbors" of $\theta$ according to $G$. 
\end{enumerate}
Moreover, by composition of the surrogate with any linear observation operator $\mathcal{H}:V \to \mathbb{R}^m$, we have that 
\begin{equation*}
\mathcal{H}(\hat{\mathcal{F}}(\theta)) = \mathcal{H}(\hat{u}_\theta) = \sum_{i\in \mathcal{I}_\theta} \alpha_i(\theta) \mathcal{H}(u_{\theta_i}).
\end{equation*}
\end{definition}

\subsection{Surrogate model's coefficients computation} \label{smcc}

The key error and stability estimates in our method to construct the surrogate model approximation follow from the model's coefficients definition. The idea behind this definition is to maintain, as far as possible, the structure of PDE in its weak form. Consequently, we consider the following estimate,   
\begin{equation}
\sup_{v\in V, \Vert v \Vert_V \leq 1} \vert  \mathcal{A}_\theta ( \hat{u}_\theta,v)- \mathcal{L}(v) \vert
    = \sup_{v\in V,\Vert v \Vert_V \leq 1}  \left\vert  \mathcal{A}_\theta \left( \sum_{i \in \mathcal{I}_\theta} \alpha_i(\theta)u_{\theta_i},v \right) -  \mathcal{L}(v) \right\vert. 
\end{equation}
Now, since by assumption $\sum_{i\in \mathcal{I}_\theta} \alpha_i(\theta) = 1$, we use the linearity of $\mathcal{A}_\theta$ to get
\begin{align*}
    \sup_{v\in V, \Vert v \Vert_V \leq 1} \vert  \mathcal{A}_\theta ( \hat{u}_\theta,v)- \mathcal{L}(v) \vert
    &= \sup_{v\in V,\Vert v \Vert_V \leq 1}  \left\vert \sum_{i \in \mathcal{I}_\theta} \alpha_i(\theta) \mathcal{A}_\theta ( u_{\theta_i},v ) -  \sum_{i\in \mathcal{I}_\theta} \alpha_i(\theta) \mathcal{L}(v) \right\vert \\
    & \leq \sup_{v\in V, \Vert v \Vert_V \leq 1}  \sum_{i\in \mathcal{I}_\theta}  \alpha_i(\theta)  \, \vert \mathcal{A}_\theta ( u_{\theta_i},v) - \mathcal{L}(v) \vert.
\end{align*}
Moreover, by the definition of weak solution $\mathcal{A}_{\theta_i} ( u_{\theta_i}, v ) = \mathcal{L}(v) \text{ for all } v \in V$, we have
\begin{align*}
    \sup_{v\in V, \Vert v \Vert_V \leq 1} \vert  \mathcal{A}_\theta ( \hat{u}_\theta,v)- \mathcal{L}(v) \vert
    &\le  \sup_{v\in V, \Vert v \Vert_V \leq 1}  \sum_{i\in \mathcal{I}_\theta}  \alpha_i(\theta) \, \vert \mathcal{A}_\theta ( u_{\theta_i},v )- \mathcal{A}_{\theta_i} ( u_{\theta_i}, v ) \vert. 
\end{align*}
Finally, recalling the continuity estimate \eqref{modulodecontinuidad}, we conclude that 
\begin{equation}\label{eq:estimateA}
    \sup_{v\in V, \Vert v \Vert_V \leq 1} \vert  \mathcal{A}_\theta ( \hat{u}_\theta,v)- \mathcal{L}(v) \vert 
     \leq  \sum_{i\in \mathcal{I}_\theta}    \alpha_i(\theta) C G(\theta, \theta_i) F(u_{\theta_i}).
\end{equation}

Estimate \eqref{eq:estimateA} is crucial in defining the surrogate's coefficients. As explained above, we aim to chose the coefficients $\alpha_i(\theta)$  such that that $\mathcal{A}_\theta(\hat{u}_\theta,v) \approx \mathcal{L}(v) \text{ for all } v \in V$. Hence, we let $\varepsilon_{\theta} = \sum_{i\in \mathcal{I}_\theta}   \alpha_i(\theta) C  G(\theta, \theta_i) F(u_{\theta_i})$ and define 
$$ 
\alpha_i(\theta)  = \frac{\varepsilon_{\theta}}{k CG(\theta, \theta_i) F(u_{\theta_i})}.
$$  
Now, since by assumption $\vert \alpha \vert= \alpha_1(\theta) + \alpha_2(\theta) + \ldots + \alpha_n(\theta) =1$, we have that
$$ 
\vert \alpha \vert = \sum_{i\in \mathcal{I}_\theta} \frac{\varepsilon_\theta}{k CG(\theta,\theta_i) F(u_{\theta_i})}=1,
$$
Because of the above argument, we have the following definition.

\begin{definition}\label{def:coeff}
We define the coefficients of the physics-based surrogate as:
\begin{enumerate}
\item[(a)] If $\theta\in\mathcal{D}$, then $\theta=\theta_j$ for some $j\in \{1,\dots,n\}$ and we let $$
\alpha_k(\theta_i) =\delta_{k,j}$$ where $\delta_{k,j}= 1$ in $k=j$ and $\delta_{k,j}= 0$ otherwise.
\item[(b)] If $\theta\notin \mathcal{D}$, then 
\begin{equation} \label{pesos}
     \alpha_i(\theta) =  \frac{1}{\displaystyle G(\theta,\theta_i) F(u_{\theta_i}) \left( \sum_{j \in \mathcal{I}_\theta} \frac{1}{ G(\theta,\theta_j) F(u_{\theta_j})} \right)  }
     \end{equation} 
for every  $i \in  \mathcal{I}_\theta$ and $\alpha_i(\theta)=0$ for every $i \in  \{1,\dots n\}\setminus\mathcal{I}_\theta$. 
\end{enumerate}
\end{definition}
 
Because of Lemma~\ref{Lemma1} in Section~\ref{Sec:sec3.5} the map from parameters to surrogate functions is continuous at design points, and both definitions of the coefficients for $\theta \in \mathcal{D}$ and $\theta \notin \mathcal{D}$ are consistent. However, we chose the above presentation to stress that we must consider the cases separately in numerical implementations due to computer rounding errors.

\subsection{Design construction} \label{sec:design}

By definition, the surrogate model depends on the given design  $\mathcal{D}$. Almost any design would produce a surrogate model. However, we require the following approximation property on the design to ensure an adequate surrogate model's performance. 

\begin{definition}\label{Def:ApproxProperty} Assume that $\mathcal{D}$ is a design on a set of parameters $\Theta\subset \mathbb{R}^d$, such that $|\mathcal{D}|=n$, $k \in \mathbb{N}$ satisfies $1<k\le n$, and $\epsilon>0$. 
We say that $\mathcal{D}$ satisfies the \emph{($\epsilon$-)approximation property} on the set of parameters $\Theta$, if for any $\theta\in\Theta$, we have that 
\begin{equation} \label{theo5eq} 
G(\theta,\theta_i) \le \epsilon \text{ for every } i\in\mathcal{I}_\theta,
\end{equation}
where $\mathcal{I}_\theta$ is the set of $k$ ``nearest neighbors'' according to $G$.
\end{definition}

\begin{remark}
    We can make ${\mathcal D}$ finer in Definition~\ref{Def:ApproxProperty} by taking $\epsilon>0$ smaller; however, this typically implies a larger number $n$ of elements in the design. 
\end{remark}

There is no algorithm to construct a design satisfying the ($\epsilon$-)approximation property for a general function $G$. Here we present a strategy that applies under additional structure assumptions on $G$. Concretely, we assume that for every $\theta_1,\theta_2 \in \Theta$ the function $G$ in (\ref{modulodecontinuidad}) can be written as
\begin{equation}\label{eq:Gstructure}
G(\theta_1,\theta_2) = g\left(d\left(f(\theta_1),f(\theta_2)\right)\right),
\end{equation}
where $g: \mathbb{R} \to \mathbb{R}^+$ is a strictly increasing function, such that $g(0)=0$, $d$ is the distance function in a metric space $(L, d)$, and  $f: \Theta \to L$ is a continuous function. The metric space $(L, d)$ may be finite or infinite dimensional, but since we assumed that $\Theta$ is finite-dimensional,  the image of $f(\Theta)$  is a subset of $(L',d)$, a finite-dimensional metric sub-space of $(L,d)$. Furthermore, we assume that the set
$f(\Theta)$ is \textit{totally bounded} \cite[p.~412]{Kreyszig}. 
 \medskip

\begin{proposition}
Assume that $G$ satisfies all assumptions in the previous paragraph, then there exists a design $\mathcal{D}$ with $|\mathcal{D}|=n$ satisfying the ($\epsilon$-)approximation property in Definition~\ref{Def:ApproxProperty} for any given $\epsilon>0$ and $1<k\le n$. 
\end{proposition}

\begin{proof}
We prove the proposition by construction.  Because $f(\Theta)$ is  totally bounded, for every $\epsilon_0>0$, there exists a set $P_{\epsilon_0} = \{ \lambda_1,\ldots, \lambda_n \}$  of elements in $L'$, called $\epsilon_0-net$, such that $L' \subset \cup_{i=1}^n B_{\epsilon_0}(\lambda_i)$.  In particular, for any $\lambda \in L'$ we could find some $\lambda_i$ in $P_{\epsilon_0}$ such that $$d(\lambda,\lambda_i) <\epsilon_0. $$ 
In this manner, if for any $\theta \in \Theta$ we let $\lambda = f(\theta)$ we have that 
\begin{equation}\label{G-epsilon}
    G(\theta,\theta_i) = g(d(f(\theta),f(\theta_i)))= g(d(\lambda,\lambda_i)) \leq g(\epsilon_0) = \tilde\epsilon,
\end{equation}
where $\theta_i = f^{-1}(\lambda_i) \in \Theta$ for some $\lambda_i\in  P_{\epsilon_0}$. 
Thus, we let $\mathcal{D} = f^{-1}(P_{\epsilon_0}) = \{\theta_1, \theta_2, \ldots , \theta_n \}$. 

By the last argument, we have that for any $\theta\in \Theta$, there is at least one theta such that equation \eqref{theo5eq} holds.  To prove the claim,  we need to 
show that the latter property holds for $k$ elements in $\mathcal{D}$ for $\tilde{\epsilon}=\epsilon$. 

Given $\lambda\in L'$ and $\lambda_i\in P_{\epsilon_0}$ as above, we let $P^k$ be a subset of $k$ elements in $P_{\epsilon_0}$, nearest to $\lambda_i$. Because the approximation property of the
$\epsilon_0-net$ the distance between $\lambda_j$ and any $\lambda_k\in P^k$ satisfy $d(\lambda_j,\lambda_k) < \lfloor \frac{k}{2}\rfloor \epsilon_0$, where $\lfloor \cdot\rfloor$ is the integer part function.
Thus, 
$$
d(\lambda_k, \lambda) \leq d(\lambda_k,\lambda_i) +d(\lambda_i,\lambda) \leq \lfloor \frac{k}{2}\rfloor \epsilon_0 + \epsilon_0 = (\lfloor \frac{k}{2}\rfloor \epsilon_0+1)\epsilon_0.$$
Therefore, 
\begin{equation} \label{Eq:eqepsilon}
G(\theta,\theta_j) = g(d(f(\theta),f(\theta_j))) = g(d(\lambda,\lambda_j)) \leq g((\lfloor \frac{k}{2}\rfloor+1)\epsilon_0) 
\end{equation} 
for every $\lambda_k$ in $P^k$. By letting $\epsilon_0 = \frac{g^{-1}(\epsilon)}{(\lfloor \frac{k}{2}\rfloor+1)} $, the proposition follows.     
\end{proof}

As a final remark, this section assumes $\Theta \subset \mathbb{R}^d$, but generalizing it to other spaces is straightforward.  

\subsection{Error Estimates}

This section presents some results regarding error estimates between the actual solution $u_\theta$ and its surrogate approximation $\hat{u}_\theta$. 

\begin{lemma}\label{lem:errorest} 
For $\theta$ given, let $u_\theta$ be the solution of \eqref{bilform} and ${\mathcal D}$ be a design with $n$ elements. Assume that $k\le n$, $\hat{u}_\theta$ is the corresponding surrogate model as in Definition~\ref{def:surrogate}, and $C>0$ is the constant on the right hand side  of \eqref{modulodecontinuidad}. Let
$$
\varepsilon_\theta= \frac{C}{\frac{1}{k}\sum_{j\in \mathcal{I}_\theta} \frac{1}{G(\theta,\theta_j) F(u_{\theta_j})}},
$$
then, we have the error estimate
\begin{equation}\label{epsilon(k)}
		\sup_{v\in V, \Vert v \Vert_V \leq 1} \vert  \mathcal{A}_\theta ( \hat{u}_\theta,v)- \mathcal{L}(v) \vert 
		\leq   \varepsilon_\theta. 
\end{equation}
	
\end{lemma}

We note that the error bound $ \varepsilon_\theta$ in Lemma~\ref{lem:errorest}, depends on the number $k$ of ``nearest neighbors" of the surrogate model. 

\begin{proof}[Proof of Lemma~\ref{lem:errorest}] We recall equation \eqref{eq:estimateA} from the derivation of the surrogate model,
	\begin{equation*}
		\sup_{v\in V, \Vert v \Vert_V \leq 1} \vert  \mathcal{A}_\theta ( \hat{u}_\theta,v)- \mathcal{L}(v) \vert 
		\leq  \sum_{i \in \mathcal{I}_\theta}   \alpha_i(\theta)  C G(\theta,\theta_i) F(u_{\theta_i}).
	\end{equation*}
By substitution of \eqref{pesos} from Definition~\ref{def:coeff} into the inequality's right hand side, we get
	\begin{align*}
		\sup_{v\in V, \Vert v \Vert_V \leq 1} \vert  \mathcal{A}_\theta ( \hat{u}_\theta,v)- \mathcal{L}(v) \vert 
		& \le \sum_{i\in \mathcal{I}_\theta} \frac{C}{\sum_{j\in \mathcal{I}_\theta} \frac{1}{G(\theta,\theta_j) F(u_{\theta_j})}}  = \frac{C}{\frac{1}{k}\sum_{j\in \mathcal{I}_\theta} \frac{1}{G(\theta,\theta_j) F(u_{\theta_j})}}. 
	\end{align*}
 Recalling $\varepsilon_\theta$'s definition, the result follows.
\end{proof}

\begin{theorem}\label{teorema1}
For $\theta$ given, let $u_\theta$ be the solution of \eqref{bilform} and $\hat{u}_\theta$ the corresponding surrogate model as in Definition~\ref{def:surrogate}. If we denote $\mathcal{A}_\theta$'s coercivity  constant by $c_\theta>0$, namely $c_\theta\Vert v \Vert^2_V \le  \mathcal{A}_\theta(v,v)$ for every $v\in V$, then 
\begin{equation*}
    \Vert \hat{u}_\theta - u_\theta \Vert_V \leq \frac{\varepsilon_\theta}{c_\theta},
\end{equation*}
where $\varepsilon_\theta$ is defined as in Lemma~\ref{lem:errorest}.
\end{theorem}  
\begin{proof} By the coercivity and 
 bi-linearity of $\mathcal{A}_\theta$, we have that 
\begin{align*}
    \Vert \hat{u}_\theta- u_\theta \Vert_V^2 &\leq \frac{1}{c_\theta} \vert \mathcal{A}_\theta (\hat{u}_\theta- u_\theta, \hat{u}_\theta- u_\theta) \vert \\
    &= \frac{1}{c_\theta} \vert \mathcal{A}_\theta (\hat{u}_\theta, \hat{u}_\theta- u_\theta) - \mathcal{A}_\theta (u_\theta, \hat{u}_\theta- u_\theta) \vert \\
    & = \frac{1}{c_\theta} \vert \mathcal{A}_\theta (\hat{u}_\theta, \hat{u}_\theta- u_\theta) - \mathcal{L}(\hat{u}_\theta- u_\theta) \vert,  \end{align*}
where for the last line we used \eqref{bilform} with  $v=\hat{u}_\theta- u_\theta$.
By letting $w=\frac{\hat{u}_\theta- u_\theta}{\Vert \hat{u}_\theta- u_\theta \Vert_V }$ in the last inequality, we get
\begin{align*}
	\Vert \hat{u}_\theta- u_\theta \Vert_V^2 
	&\le  \frac{1}{c_\theta} \Vert \hat{u}_\theta - u_\theta \Vert_V \vert \mathcal{A}_\theta (\hat{u}_\theta, w) - \mathcal{L}(w) \vert \\
	& \leq  \frac{1}{c_\theta} \Vert \hat{u}_\theta - u_\theta \Vert_V \sup_{w \in V, \Vert w\Vert_V \leq 1} \vert \mathcal{A}_\theta (\hat{u}_\theta, w) - \mathcal{L}(w) \vert.
\end{align*}
Finally, by Lemma~\ref{lem:errorest}, we conclude that
$\Vert \hat{u}_\theta- u_\theta \Vert_V \leq \varepsilon_\theta/c_\theta$ as claimed. 
\end{proof}

The following lemma is crucial to our method because it shows we can make the surrogate's error small by making the design $\mathcal{D}$ finer. 

\begin{theorem} \label{theorem5}
Assume that $\mathcal{D}$ is a design in $\Theta$ satisfying the approximation property in Definition~\ref{Def:ApproxProperty} for $\epsilon > 0$, $\varepsilon_\theta$ is the error bound in Lemma~\ref{lem:errorest} for any given $\theta\in \Theta$, $C>0$, and $F: V\to {\mathbb R}$, where the latter two are the constant and function appearing on the right hand side of \eqref{modulodecontinuidad}. Let $C_F = \sup_{\theta' \in \Theta} F(u_{\theta'})$ and notice that by assumption $C_F < \infty$. 
Then, 
\begin{equation*}
\varepsilon_\theta \leq CC_F\epsilon.
\end{equation*}
\end{theorem}

\begin{proof}
Assume $\mathcal{D} = \{ \theta_1,\ldots,\theta_n \}$ is the given design and $\theta\in \Theta$. Because of the approximation property, 
$G(\theta, \theta_j) \leq \epsilon $
for every $j  \in \mathcal{I}_\theta$, and since $F$ and $G$ are both positive, we have that
\begin{equation*} 
G(\theta,\theta_j)F(u_{\theta_j})\leq \epsilon F(u_{\theta_j})\leq \epsilon C_F ,\end{equation*}
for all $j\in \mathcal{I}_\theta$. Thus 
\begin{equation*}\frac{1}{G(\theta,\theta_j)F(u_{\theta_j})} \geq \frac{1}{C_F }\frac{1}{\epsilon}. \end{equation*}
By averaging over $j$ in $\mathcal{I}_\theta$, last equation yields
\begin{equation*} \frac{1}{k} \sum_{j\in \mathcal{I}_\theta} \frac{1}{G(\theta,\theta_j)F(u_{\theta_j})} \geq \frac{1}{k}\sum_{j\in \mathcal{I}_\theta} \frac{1}{C_F }\frac{1}{\epsilon} = \frac{1}{C_F }\frac{1}{\epsilon}. \end{equation*}
Finally, we conclude that  
\begin{equation*} \varepsilon_\theta = \frac{C}{\frac{1}{k} \sum_{j\in \mathcal{I}_\theta} \frac{1}{G(\theta,\theta_j)F(u_{\theta_j})}} \leq CC_F\epsilon, \end{equation*}
by recalling the definition of $\varepsilon_0$ in Lemma~\ref{lem:errorest}.
\end{proof}


\begin{corollary} \label{cor:surr_approx}
Under the assumption of Theorem~\ref{theorem5}, let $\theta\in\Theta$ and $\hat{u}_\theta$ be the corresponding surrogate solution as in Definition~\ref{def:surrogate}. Then, for any $\epsilon>0$, we have that:
\begin{enumerate}
\item[(a)] The error in the weak equation is controlled by
$$ \sup_{v\in V, \Vert v \Vert \leq 1} \vert \mathcal{A}_\theta(\hat{u}_\theta,v)-\mathcal{L}(v)\vert \leq CC_F\epsilon.$$
\item[(b)]\label{cor:surr_approx_b} The error between the actual weak solution $u_\theta$ and its surrogate is bounded by  
$$
\Vert \hat{u}_\theta - u_\theta \Vert_V \leq  \frac{C C_F}{c_\theta} \epsilon .
$$
\end{enumerate}

\end{corollary}

\begin{proof}
The proofs of items (a) and (b) follow from Theorem \ref{theorem5} combined with Lemma~\ref{lem:errorest} and \ref{teorema1}, respectively.
\end{proof}

The first item in the last corollary implies that the underlying physics of the problem is approximately maintained, up to some error, as we aimed in our surrogate model derivation. The second item shows the convergence of the surrogate model approximation to the real solution as we take finer designs $\mathcal{D}$.
The latter result is crucial for the next section, where we study the existence and consistency between surrogate and real posterior distributions.

\subsection{Existence and consistency of the posterior distribution} \label{Sec:sec3.5}

This section presents the results on the existence
of the posterior distributions for the exact and surrogate models and their consistency. 

We begin by defining the Bayesian inference problem notation. Let $\mathcal{M}$ be the following regression model
\begin{equation}\label{modelootravez1}
  y_i =  \mathcal{H}_i(\mathcal{F}(\theta))+ \epsilon_i \quad \mbox{ with } \epsilon_i \sim N(0, \sigma^2),
\end{equation} 
for $i=1,\ldots,m$. The corresponding likelihood and marginal distributions are 
\begin{equation} \label{likelihood-fm}
P_{Y\mid \Theta} (y\mid \theta) = (2 \pi \sigma^2 )^{-m/2} \exp \left( -\frac{1}{2\sigma^2} \sum_{i=1}^m (y_i-\mathcal{H}_i(\mathcal{F}(\theta)))^2 \right),
\end{equation}
and
\begin{equation} \label{constnorm-fm}
P_{Y}(y) = \int P_{Y \mid \Theta}(y\mid \theta ) \pi(\theta) \, d\theta,
\end{equation} 
where $\pi(\theta)$ is the a priori distribution. 
By Bayes theorem, the posterior distribution of $\mathcal{M}$ for the parameter $\theta$ given the observations $y$, is  
\begin{equation}\label{posterior-fm}
P_{\Theta \mid Y}(\theta\mid y) = \frac{P_{Y\mid \Theta} (y\mid \theta) \pi(\theta)}{P_{Y}(y)}.
\end{equation}

By Definition~\ref{def:surrogate}, the corresponding surrogate regression model $\hat{\mathcal{M}}$ is \begin{equation}\label{modelootravez2}
  y_i = \mathcal{H}_i(\hat{\mathcal{F}}(\theta))+ \epsilon_i \quad \mbox{ with } \epsilon_i \sim N(0, \sigma^2),
\end{equation} 
for $i=1,\ldots,m$. Surrogate model's likelihood $\hat{P}_{Y\mid \Theta} (y\mid \theta)$, and marginal distribution $\hat{P}_{Y}(y)$ follow as above with straightforward changes. 
Hence, the surrogate's posterior distribution is
\begin{equation} \label{posterior-sm} 
\hat{P}_{\Theta \mid Y}(\theta\mid y) = \frac{\hat{P}_{Y\mid \Theta} (y\mid \theta) \pi(\theta)}{\hat{P}_{Y}(y)}.
\end{equation}

The existence and well-posedness of the posterior distributions follow from the regularity conditions on the map from parameters to observation, namely $\mathcal{H}\circ \mathcal{F}$. By the assumption in Section~\ref{sec:general}, the observation operator $\mathcal{H}$ is continuous; therefore, we are only concerned with the forward map regularity properties.
First, we analyze the exact model. The following lemma shows the continuity of the forward map ${\mathcal F}$.  

\begin{lemma} \label{CFM}
Under the assumptions in Section~\ref{sec:PBSMM}, the forward map $\mathcal{F}:\Theta \to V$ is continuous.
\end{lemma}

\begin{proof} 
Let $\theta,\theta'\in\Theta$. By the coercivity of the bi-linear form $\mathcal{A}_\theta$ we have that
\begin{equation*}
\Vert \mathcal{F}(\theta)-\mathcal{F}(\theta') \Vert_V^2 
= \Vert u_\theta-u_{\theta'} \Vert_V^2 
\leq \frac{1}{c_\theta} \vert \mathcal{A}_\theta(u_\theta-u_{\theta'},u_\theta-u_{\theta'}) \vert, 
\end{equation*}
where $c_\theta$ is the coercivity constant  $\mathcal{A}_\theta$.
Let  $v = \frac{u_\theta-u_{\theta'}}{\Vert u_\theta-u_{\theta'} \Vert_V}$, and use the bi-linearity of $\mathcal{A}_\theta$ to get 
\begin{equation}\label{eq:X1}
\Vert u_\theta-u_{\theta'} \Vert_V^2 
\le \frac{\Vert u_\theta-u_{\theta'} \Vert_V}{c_\theta} \left\vert \mathcal{A}_\theta \left( u_\theta, v\right) - \mathcal{A}_\theta\left( u_{\theta'}, v \right) \right\vert.
\end{equation}
Since $u_\theta$ and $u_{\theta'}$ are both solution, we have that 
$\mathcal{A}_\theta\left( u_{\theta'}, v \right) = \mathcal{L}(v) = \mathcal{A}_\theta\left( u_{\theta'}, v \right)$ for every $v\in V$ . By substitution of last equation into \eqref{eq:X1} and simplifying terms, we obtain 
$$
\Vert u_\theta-u_{\theta'} \Vert_V 
\leq \frac{1}{c_\theta} \vert \mathcal{A}_{\theta'} (u_{\theta'},v) - \mathcal{A}_\theta(u_{\theta'},v) \vert 
\leq \frac{1}{c_\theta} \sup_{w\in V , \Vert w \Vert_V \leq 1} \vert \mathcal{A}_{\theta'} (u_{\theta'},w) - \mathcal{A}_\theta(u_{\theta'},w) \vert. 
$$
In the last inequality, we have used $\Vert v\Vert_V =1$. Now, by the uniform ellipticity of $\mathcal{A}_\theta$, we have that $\tilde{c}= \sup_{\theta \in \Theta} \frac{1}{c_\theta} < \infty$. Recalling the continuity estimate \eqref{modulodecontinuidad} we get 
$$
\Vert \mathcal{F}(\theta)-\mathcal{F}(\theta') \Vert_V=
\Vert u_\theta-u_{\theta'} \Vert_V 
\leq \tilde{c}C G(\theta,\theta') F(u_\theta).
$$
Thus, since by assumption $F:V\to {\mathbb R}$ is bounded, the continuity of ${\mathcal F}$ follows from the continuity of $G$, as claimed. 
\end{proof}

Results like Lemma~\ref{CFM} are typical in the study of elliptic PDEs \cite[Chapter~14]{Gilbarg}. We include detailed proof using our particular assumptions and notation for completeness only. Now, as a consequence of Lemma~\ref{CFM} and Lemmas 1 and 2 from \cite{christenperezg}, the posterior measures $P_{\Theta\mid Y}$ exists and is well-defined as a probability measure.

\medskip

Regarding the surrogate model, the following lemma shows continuity at the design points. 
\begin{lemma} \label{Lemma1}
 The function $\theta\to \hat{u}_\theta$ is continuous at every point of $\mathcal{D}$. 
That is, $\hat{u}_\theta \to u_{\theta_i}$ as $\theta \to \theta_i$ for any  $\theta_i \in \mathcal{D}$.
\end{lemma}

\begin{proof}When $\theta \to \theta_i$ for any $\theta_i \in \mathcal{D}$ and $i = 1,\ldots,n$, then $G(\theta,\theta_i) \to 0$. Then
\begin{align*}
\lim_{\theta \to \theta_i} \alpha_i(\theta) &= \lim_{\theta \to \theta_i}  \frac{1}{\displaystyle G(\theta,\theta_i) F(u_{\theta_i}) \left( \sum_{j \in \mathcal{I}_\theta} \frac{1}{ G(\theta,\theta_j) F(u_{\theta_j})} \right)  } \\
&= \lim_{\theta \to \theta_i} \frac{1}{\displaystyle G(\theta,\theta_i) F(u_{\theta_i}) \left( \frac{1}{G(\theta,\theta_i)F(u_{\theta_i})} + \sum_{j \in \mathcal{I}_\theta, j\neq i} \frac{1}{ G(\theta,\theta_j) F(u_{\theta_j})} \right) } \\
&= \lim_{\theta \to \theta_i} \frac{1}{ \displaystyle  1 + G(\theta,\theta_i)F(u_{\theta_i}) \left(\sum_{j \in \mathcal{I}_\theta, j\neq i} \frac{1}{ G(\theta,\theta_j) F(u_{\theta_j})}  \right)}.
\end{align*}
Now,  since $\sum_{j \in \mathcal{I}_\theta, j\neq i} \frac{1}{ G(\theta,\theta_j) F(u_{\theta_j})} < \infty$ and $\lim_{x\to 0} \frac{1}{1+bx} = 1$ with $b$ constant we 
 have that 
 \begin{equation*}
  \lim_{\theta \to \theta_i} \alpha_i(\theta) = 1.
 \end{equation*}
 Hence, $\alpha_j(\theta)\to 0$ for every $j\in \mathcal{I}_\theta$ such that $j\neq i$ and consequently $\hat{u}_\theta \to u_{\theta_i}$ as $\theta \to \theta_i$.
\end{proof}

However, the continuity can no be extended all points in $\Theta$. Indeed, for any given design ${\mathcal D}$, there exist values $\theta\in\Theta$ where the set $\mathcal{I}_\theta$  of  $k$ ``nearest neighbors" is not uniquely defined. For instance, if $\Theta=[0,5]$ and ${\mathcal D}=\{1,2,3,4\}$, the sets of $k=3$  ``nearest neighbors" of $\theta=2.5$  are $\{1,2,3\}$ and $\{2,3,4\}$. 
Therefore, the surrogate jumps at these points, and the $\hat{\mathcal F}:\Theta\to V$ is not continuous by construction. The latter is an interesting property because it shows that continuity in the surrogate models' approximation is unnecessary to yield good performance in the inference process.  
In general, the existence of posterior distributions follow via milder regularity properties than continuity. Measurability properties on conditional densities are enough (for instance, see \cite[page 6]{Ghosal2017}). In the current setting, this amounts to the surrogate forward map $\hat{\mathcal F}:\Theta\to V$ being measurable, which by its definition indeed holds. Therefore, the posterior distribution for the surrogate model $\hat{P}_{\Theta \mid Y}$ exists and is well-defined.

\begin{proposition}\label{prop:consistency}
Let $P_{\Theta \mid Y}$ and  $\hat{P}_{\Theta \mid Y}$ be the posterior probability densities from the exact and surrogate models defined as above. Then, for any $\epsilon >0$, there exists a constant $K>0$ such that 
\begin{equation}\label{eqn:post_convergence}
\Vert P_{\Theta\mid Y} - \hat{P}_{\Theta\mid Y} \Vert_{TV} \leq \frac{K}{P_Y(y)} \epsilon ,
\end{equation}
where $P_Y(y)$ is the normalization constant in \eqref{constnorm-fm}.
\end{proposition}

\begin{proof} Let $\epsilon>0$ as given in the assumptions, and a $\mathcal{D}$ a design satisfying $\epsilon$-approximation property. By Corollary \ref{cor:surr_approx}, we have\begin{equation}\label{eq:errorModels}\Vert \hat{u}_\theta - u_\theta \Vert_V \leq \tilde{c}CC_F \epsilon.\end{equation}where $\tilde{c} = \sup_{\theta \in \Theta} \frac{1}{ c_\theta} < \infty$, $C_F=\sup_{\theta' \in \Theta} F(u_{\theta'}) < \infty$, and $C$ is the constant on the right hand side of \eqref{modulodecontinuidad}. 
By \eqref{eq:errorModels} and Theorem 2 in \cite{christenperezg}, we conclude \eqref{eqn:post_convergence} as claimed. Note, the assumptions in Theorem 2  in \cite{christenperezg}  can be relaxed from continuity to existence and measurability of the surrogate's forward map.
\end{proof}

Proposition~\ref{prop:consistency} shows that the posterior's order of convergence in TV norms is inherited from the FM's order of convergence—a well-known property in BUQ problems. 

\subsection{Algorithm implementation}
The algorithmic implementation of the physically based surrogate models is relatively straightforward. It consists of a preprocessing, presented in Algorithm~\ref{alg:preproc}, and the surrogate model evaluation, presented in Algorithm~\ref{alg:MCMC}. The latter algorithm is the only part that enters open-loop applications such as Bayesian inference through an MCMC exploration algorithm. 

\begin{algorithm} [H] 
\caption{Preprocessing}
\label{alg:preproc}
\begin{algorithmic}
\STATE{\textbf{Data:} a design $\mathcal{D} = \{ \theta_1,\ldots,\theta_n \}$ in $\Theta$}
\FOR{$\theta_i \in \mathcal{D}$} 
\STATE{Evaluate FM $u_{\theta_i} = {\mathcal F}(\theta_i)$}
\STATE{Compute $\mathcal{H}_j(u_{\theta_i})$ for $j=1,\ldots,m$}
\STATE{Set $A_i = \{ \mathcal{H}_j(u_{\theta_i}) \}_{j=1}^m $}
\ENDFOR
\RETURN $A = \{ A_1,\ldots, A_n \}$
\end{algorithmic}
\end{algorithm}

The preprocessing step is the most computationally expensive part because each surrogate element computation involves an exact forward map numerical evaluation. For elliptic PDE models, a typical choice would be 
the finite element method (like the one in our examples below). In addition, if the surrogate elements $\{u_{\theta_i}\}_{i=1}^{n}$ are given by large vectors, storing only $\{{\mathcal H}({\mathcal F}(\theta_i))\}_{i=1}^n$  saves memory space and implies less computations. 

\begin{algorithm} [H]
\caption{Surrogate-Model evaluation}
\label{alg:MCMC}
\begin{algorithmic}
\STATE{\textbf{Data:} $\theta \in \Theta$, $A$, $G$, $k$ and $\eta$}
\IF{$G(\theta,\theta_i) < \eta$ for some $\theta_i \in \mathcal{D}$} 
\STATE{$S = A_i$}
\ELSE
\STATE{Compute $\mathcal{I}_\theta$ such that $\vert \mathcal{I}_\theta \vert =k$}
\FOR{$i \in \mathcal{I}_\theta$}
\STATE{Compute $\alpha_i(\theta)$}
\ENDFOR
\STATE{$S = \sum_{i \in \mathcal{I}_\theta} \alpha_i(\theta) A_i$}
\ENDIF
\RETURN $S$
\end{algorithmic}
\end{algorithm}

The parameter $\eta$ sets a neighborhood around each element of $\mathcal{D}$ to avoid small number division errors in the coefficient computations. We used $\eta = 10^{-8}$ in examples of section \ref{sec:examples}. The computational cost of the surrogate model evaluations is small, at least when the parameters are low dimensional spaces. For larger dimensional parameter spaces computing the  $k$ of ``nearest neighbors"  can be computationally expensive and challenge the method's efficiency.

The following section presents some examples to test the efficiency and efficacy of our Physics-Based Surrogate Model for some low-dimension regression examples. We used a finite element method for the elliptic PDE numerical solution in the preprocessing algorithm, and FreeFem \cite{freefem} for its implementation. For the MCMC exploration, we used t-walk algorithm \cite{twalk} implemented in Python 3. Moreover, separate implementations are possible for preprocessing step and surrogate evaluations since the complete algorithm does not require online communication between both parts. Finally, we note that examples in Section~\ref{sec:examples} were implemented on a personal computer with a 4.70 GHz processor.


\section{Examples} \label{sec:examples}

In this section weIn this section, we present examples of the method for a problem in electrical impedance tomography (EIT), a medical imaging technique in which an image of the conductivity (or permittivity) of part of a body is determined from electrical surface measurements. In the mathematical literature, EIT inverse problem is also known as Calderón's problem \cite{Calderon-2006}. The examples below demonstrate the surrogate models' efficiency and efficacy for one or two parameter inverse problems. We use synthetic data and, when available, the exact analytical solution to asses the method's performance.

\subsection{EIT problem on a disk}\label{EIT:disk}
Let $B \subset \mathbb{R}^2$ be the unit disk, denote its boundary by $\partial B$, and assume $\lambda: B \to \mathbb{R}$ and $f: \partial B \to \mathbb{R}$ are real-valued functions. The direct problem is to find $u$ such that
\begin{equation}\label{problem1}
    \mbox{div} (\lambda \nabla u) = 0 \quad \mbox{ in } B, \quad \lambda \frac{\partial u}{\partial n} = f \quad \mbox{ on } \partial B.
\end{equation}  
We need further hypotheses to ensure the existence and uniqueness of solutions to \eqref{problem1}. First, $\lambda \in L^\infty (B)$ and there exists $\lambda_0>0 $ such that $\lambda(x) > \lambda_0$ for every $x\in B$, namely problem~\eqref{problem1} is uniformly elliptic.  Second, by the divergence theorem, solutions to \eqref{problem1} exist if and only if 
\begin{equation}\label{consistent}
\int_{\partial B} f \, ds = 0.
\end{equation}
Third, solutions are unique up to an additive constant. To get rid of this degeneracy, we are concerned only with solutions $u$ such that 
$$ \int_{\partial B} u \, ds = 0.$$
In view of the latter assumptions, we introduce the following functional spaces
$$
H^\frac12_0(\partial B) = \left\{ f \in H^\frac12(\partial B) \, : \, \int_{\partial B} f \, ds = 0 \right\}
\text{ and }
 H_0^1 (B) = \left\{ u \in H^1(B) \, :  \, \int_{\partial B} u \, ds = 0 \right\}, $$
where the norm for the latter space is given by   $\Vert u \Vert_{H^1_0(B)} = \Vert \nabla u \Vert_{L^2}$ (for instance, see \cite{Gilbarg}). 

The weak form of problem~\eqref{problem1} is given by 
$$ {\mathcal A}(u,v):=\int_B \lambda \nabla u \cdot \nabla v \, dx = \int_{\partial B} v f \, ds =: {\mathcal L}(v)\quad \mbox{ for all } v \in H_0^1(B) .$$
With above notations, ${\mathcal A}:H^1_0(B)\times H^1_0(B)\to {\mathbb R}$ is a continuous and coercive bi-linear form and ${\mathcal L}: H_0^1 (B) \to {\mathbb R}$ is a linear functional. Therefore, existence and uniqueness follows from the Lax-Milgram theorem.

\subsection*{Explicit exact solutions} \label{asftfm}
We derive an explicit exact solution to problem~\eqref{problem1} in the case where the conductivity has the form 
\begin{equation} \label{cond1}
\lambda (x) = \left\{ \begin{array}{ll}
1     &  \mbox{ if }R<\vert x \vert <1, \\
1+\rho     & \mbox{ if } \vert x \vert < R.  
\end{array} \right.
\end{equation} 
That is a disk of radius one of constant conductivity $\lambda = 1 $ with a disk-centered inclusion of radius $0<R<1$ and conductivity $\lambda = 1 + \rho$. 
In this case, the explicit solution of $u$ on  $\partial B$  is given by 
\begin{equation}\label{eq:explictsol}
u(1,\theta) = \sum_{n > 0} \frac{\alpha-R^{2^\vert n \vert}}{\alpha+R^{2\vert n \vert}} \frac{f_n}{\vert n \vert} \cos(n\theta), \quad \theta \in [0, 2\pi], 
\end{equation}
where $\alpha = 1+2/\rho$ and $f_n$ is the cosine series coefficient (see \cite{Kirsch}).  For the examples below,  we assume that $f(\theta) = \cos(4 \theta)$, and 
by substituting the latter into \eqref{eq:explictsol}, we get 
\begin{equation}\label{eq:explicit}
u(1,\theta) = \frac{1}{4}\cdot \frac{\alpha- R^{8}}{\alpha + R^{8}}  \cos(4\theta), \quad \theta \in [0,2\pi].
\end{equation}

We use \eqref{eq:explicit} in the first two inverse problems presented below to test the Physic-Based surrogate model efficacy to obtain acceptable posterior estimates. 

\subsection{Inference on the radious and conductivity value}
This section presents two examples of Bayesian inference to solve an inverse problem for one unknown parameter. In both cases, the exact forward map is available to assess the efficacy of the Physics-Based surrogate model. A comparison with a brute force approach with the numerical FEM implementation measures the surrogate's model efficiency.

\subsubsection{Inferring the conductivity value} \label{ICV}

\subsubsection*{Inverse Problem}
The first inverse problem is inferring a disk radial inclusion's conductivity value $\rho$ from a series of potential measurements $u(x)$ on $m$ boundary points $\{x_i\}_{i=1}^m\subset \partial B$. We assume that the inclusion radius $R$ is known.  To stress that the unknown parameter is $\rho$, we use the notation $\lambda_\rho$ for the conductivity function $\lambda$ in this example.

\subsubsection*{Bilinear's Form Continuity Estimate} In this case, the bilinear form  is 
$$ \mathcal{A}_\rho (u,v) = \int_B \lambda_\rho \nabla u \cdot \nabla v \, dx$$
for every  $u,v \in H_0^1 (B)$, and
the continuity estimate \eqref{modulodecontinuidad}
becomes
\begin{align*}
    \sup_{v\in V, \Vert v \Vert_V \leq 1}  \vert \mathcal{A}_\rho ( u ,v )- \mathcal{A}_{\rho'} ( u, v ) \vert &\leq  \sup_{v\in V, \Vert v \Vert_V \leq 1}   \int_B \vert\lambda_\rho-\lambda_{\rho'}\vert \vert \nabla u \cdot \nabla v \vert \, ds \\
    & \leq \Vert \lambda_\rho - \lambda_{\rho'} \Vert_{L^\infty (B)} \Vert \nabla u \Vert_{L^2(B)} \sup_{v\in V, \Vert v \Vert_V \leq 1}  \Vert \nabla v \Vert_{L^2(B)} \\
    & \leq \Vert \lambda_\rho - \lambda_{\rho'} \Vert_{L^\infty (B)} \Vert \nabla u \Vert_{L^2(B)}, 
\end{align*}
where we used Cauchy-Schwarz's inequality for the second line. Comparing latter inequality with \eqref{modulodecontinuidad}, we get 
$$
C=1,\quad G(\rho,\rho') = \vert \rho - \rho' \vert, \text{ and } F(u)= \Vert \nabla u \Vert_{L^2(B)}.
$$
where we have used that
$\Vert \lambda_\rho - \lambda_{\rho'} \Vert_{L^\infty (B)} = \vert \rho - \rho' \vert$
because of equation \eqref{cond1}.

\subsubsection*{Design}
Let $\Theta = [0,10]$ be the admissible parameters set. We follow the procedure in Section~\ref{sec:design} to construct a design ${\mathcal D}$  satisfying the $\varepsilon$-approximation 
property on $\Theta$. Because of $G(\rho,\rho') = \vert \rho-\rho'\vert$, we identify $g(x)=x$, $d(a,b)=\vert a-b\vert$ and $f(\rho) = \rho$ in the structure hypothesis \eqref{eq:Gstructure}.  The metric space is $L={\mathbb R}$ equipped with the Euclidean distance $d$, and $f(\Theta)=[0,10] \subset L'=L$. As $[0,10]$ is a bounded set in ${\mathbb R}$ it is totally bounded, i.e., for every $\varepsilon >0$, there exists a finite cover such that the radius of each element cover is at most $\varepsilon$. Since $L'$ is a segment on the real line, the construction of the design ${\mathcal D}$ is straightforward; let $\mathcal{D} = \{ 10k/2^{l}\}_{k=0}^{2^l}$. Consequently, $\mathcal{D}$ will satisfy the $\varepsilon$-approximation of Definition~\ref{Def:ApproxProperty} provided $l \ge  \ln\left(\frac{1-\varepsilon}{\varepsilon}\right) / \ln(2)$. Notice that the number of elements in the design is $n=2^{l}+1$, which is also the number of exact forward map evaluations in the preprocessing step.   

\subsubsection*{Surrogate Model}

Given the design ${\mathcal D} = \{ \rho_1, \ldots, \rho_n \}$ as above, the surrogate approximation is given by 
$$
\hat{u}_\rho = \sum_{i\in \mathcal{I}_\rho} \alpha_i(\rho) u_{\rho_i},
$$
where 
$$  \alpha_i(\rho) =  \frac{1}{\displaystyle \vert \rho-\rho_i \vert \Vert \nabla u_{\rho_i} \Vert_{L^2(B)} \left( \sum_{j\in \mathcal{I}_\rho} \frac{1}{ \vert \rho-\rho_j \vert \Vert \nabla u_{\rho_i} \Vert_{L^2(B)}} \right) },$$
and constructing the $k$-nearest neighbors index set $\mathcal{I}_\rho$ is straightforward.   

\subsubsection*{Error Estimates}
From above definitions, the  bilinear's form  error is estimated by 
\begin{align*}
    \sup_{v\in V, \Vert v \Vert_V \leq 1} \vert  \mathcal{A}_\rho ( \hat{u}_\rho,v)- \mathcal{L}(v) \vert &\leq  \frac{1}{\displaystyle \frac{1}{k}  \sum_{j\in \mathcal{I}_\rho} \frac{1}{\vert \rho - \rho_j \vert \Vert \nabla u_{\rho_j} \Vert_{L^2(B)}} } = \varepsilon_\rho.
\end{align*}

\subsubsection*{Numerical Results}
For the numerical test we let $R= 0.85$ and take $m=10$ measurement points $x_i = (\cos(2\pi(i/10)), \sin(2\pi(i/10)))$ for $i \in\{ 0\ldots, 9\}$. We implemented the exact forward numerical FEM scheme in Freefem on a triangular mesh with $56,206$ elements and  $111,610$ vertices. The latter implementation generated the synthetic measurement for $\rho=3.2$ and computed $F(u_\rho)$ for the surrogate coefficients. In addition, we let $\sigma = 0.01$ as the measurement noise standard deviation in \eqref{modelootravez1}.
We assume a uniform apriori distribution supported on $[0,10]$ and implemented MCMC algorithm t-walk to get $800,000$ samples from the posterior distribution. Finally, we considered $k=2$ in the nearest neighbors' computation.

Table~\ref{table:1} compares preprocessing, MCMC exploration, and total times for different numbers $n$ of surrogate elements. The last line reports the sum of all executions for every value of $n$.

\begin{table}[H]
\centering
\caption{Column $n$ is the number of solutions in the surrogate model, column a) is the time to compute all surrogate elements, $u_{\theta_i}$, in FreeFem, column  b) is the time required for the MCMC exploration with $800,000$ iterations, column  c) is the total time required, and column $\varepsilon_\rho$ is the error estimate. The sum in the last line is the time required to compute all models. }
\begin{tabular}{c c c c c} 
 \hline
 \textbf{n} & \textbf{a)} & \textbf{b)} & \textbf{c)} & \textbf{$\varepsilon_\rho$} \\  
 \hline
$n = 5$ & 0.769 min  & 2.2 min  & 2.969 min & $30/2^2$ \\
$n = 9$ & 1.369 min & 2.4 min & 3.769 min  & $30/2^3$\\
$n = 17$ & 2.567 min & 3.3 min & 5.867 min  & $30/2^4$ \\
$n = 33$  & 4.970 min & 3.5 min & 8.470 min & $30/2^5$ \\
 \hline
 & & Sum = & 21.075 min &
\end{tabular}
\label{table:1}
\end{table}

The brute force approach of using the exact forward map for $800,000$ iterations would take approximately 78 days. Therefore, the time to sample the posterior distribution with $n=33$ surrogate elements is about $10^4$ times faster. Even performing all experiments for different values of $n$ is four orders of magnitude faster than to use the exact forward map.

\begin{figure}[!ht]
\centering
\begin{subfigure}{0.49\textwidth}
\centering
\includegraphics[width=\textwidth]{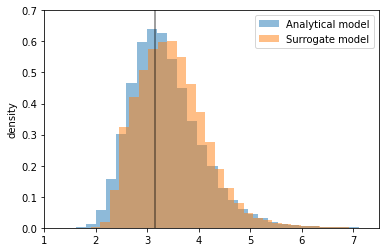}
\caption[]%
{{\small $n=5$}}    
\label{fig:mean and std of net14}
\end{subfigure}
\hfill
\begin{subfigure}{0.49\textwidth}  
\centering 
\includegraphics[width=\textwidth]{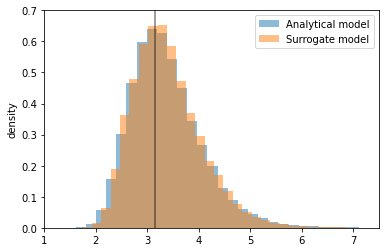}
\caption[]%
{{\small $n=9$}}    
\label{fig:mean and std of net24}
\end{subfigure}
\vskip\baselineskip
\begin{subfigure}{0.49\textwidth}   
\centering 
\includegraphics[width=\textwidth]{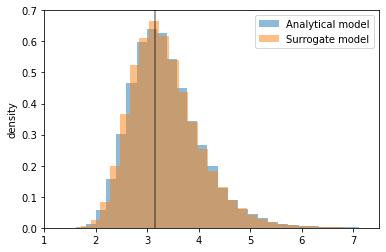}
\caption[]%
{{\small $n=17$}}    
\end{subfigure}
\hfill
\begin{subfigure}{0.49\textwidth}   
\centering 
\includegraphics[width=\textwidth]{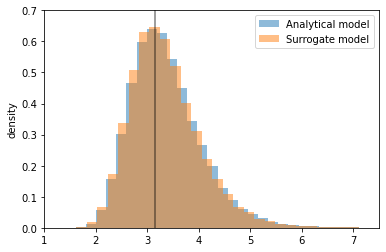}
\caption[]%
{{\small $n=33$}}    
\label{fig:mean and std of net44}
\end{subfigure}
\caption[]
{\small  \label{fig:infer_rho}  Histograms of the MCMC samples for the posterior distribution with four designs, increasing the number of elements in our design $\mathcal{D}$.  A binary partition of $[0,10] = \Theta$ was used. Convergence to exact posterior is observed with few model evaluations. The vertical line is the true value of $\rho$.} 
\label{fig:mean and std of nets}
\end{figure}

Figure~\ref{fig:infer_rho} presents the posterior distributions for different numbers of surrogate elements $n$  and a burn-in of
10,000 samples, after inspection of the MCMC output in each case. We included a sample of the posterior distribution obtained after $800,000$ iterations with the exact analytical solution \eqref{eq:explicit}  for comparison. 
In Figure~\ref{fig:infer_rho} (a) to (d), we can see how $\hat{P}_{\Theta\mid Y}$ approaches to the true posterior distribution $P_{\Theta\mid Y}$ as $n$ increases. Some differences between the posteriors sampled with the analytical solution and the surrogate model are still present for $n=5$. For $n=17$, similarities between both posterior distributions are remarkable.   Proposition~\ref{prop:consistency} assures the convergence of $\hat{P}_{\Theta\mid Y}$ to the true posterior $P_{\Theta\mid Y}$, as $n$ increases.

\subsubsection{Inferring the Conductivity Support}\label{exmaple:ICS}

\subsubsection*{Inverse Problem} The second inverse problem is inferring a centered disk inclusion's radius $R$  from a series of potential measurements $u(x)$ on $m$ boundary points $\{x_i\}_{i=1}^m\subset \partial B$. We assume that the inclusion conductivity $\rho$ is known. To stress that the unknown parameter is $R$, we use the notation $\lambda_R$ for the conductivity function $\lambda$ in this example.

\subsubsection*{Bilinear's Form Continuity Estimate} In this case, the bilinear form  is 
$$ \mathcal{A}_R (u,v) = \int_B \lambda_R \nabla u \cdot \nabla v \, dx$$
for every  $u,v \in H_0^1 (B)$, and
the continuity estimate \eqref{modulodecontinuidad}
becomes
\begin{align*}
    \sup_{v\in V, \Vert v \Vert_V \leq 1}  \vert \mathcal{A}_R ( u,v )- \mathcal{A}_{R'} ( u, v ) \vert &\leq  \sup_{v\in V, \Vert v \Vert_V \leq 1}   \int_B \vert\lambda_R-\lambda_{R'}\vert \vert \nabla u \cdot \nabla v \vert \, dx \\
    & \leq  \Vert \lambda_R - \lambda_{R'} \Vert_{L^4(B)} \Vert \nabla u \Vert_{L^4(B)} \sup_{v\in V, \Vert v \Vert_V \leq 1}  \Vert \nabla v \Vert_{L^2(B)} \\
    & = \Vert \lambda_R - \lambda_{R_i} \Vert_{L^4(B)} \Vert \nabla u_{R_i} \Vert_{L^4(B)} \\
    & \leq (2\pi)^{1/4} \rho \vert R-R' \vert^{1/4} \Vert \nabla u_{R_i} \Vert_{L^4(B)} .
\end{align*}
where we used H\"older's inequality twice and equation \eqref{cond1} for the second and fourth lines, respectively. Comparing latter inequality with \eqref{modulodecontinuidad}, we get 
$$
C=(2\pi)^{1/4}\rho,\quad G(R,R') = \vert R-R' \vert^{1/4}, \text{ and } F(u)= \Vert \nabla u \Vert_{L^4(B)}.
$$

\subsubsection*{Design}
Let $\Theta = [0,1]$ be the admissible parameters set. We follow the procedure in Section~\ref{sec:design} to construct a design ${\mathcal D}$  satisfying the $\varepsilon$-approximation 
property on $\Theta$. Because of $G(R,R')=g(d(f(R),f(R')))  = \vert R-R'\vert^{1/4}$, we identify $g(x)=x^{1/4}$, $d(a,b)=\vert a-b\vert$ and $f(R) = R$ in the structure hypothesis \eqref{eq:Gstructure}.  The metric space is $L={\mathbb R}$ equipped with the Euclidean distance $d$, and $f(\Theta)=[0,1] \subset L' = L$. As $[0,1]$ is a bounded set in ${\mathbb R}$ it is totally bounded, and the construction of the design ${\mathcal D}$ is straightforward; let $\mathcal{D} = \{ k/2^{l}\}_{k=0}^{2^l}$. Consequently, $\mathcal{D}$ will satisfy the $\varepsilon$-approximation of Definition~\ref{Def:ApproxProperty} provided $l \ge  \ln\left(\frac{1-\varepsilon}{\varepsilon}\right) / \ln(2)$. The number of elements in the design is $n=2^{l}+1$.   

\subsubsection*{Surrogate Model}

Given the design ${\mathcal D} = \{ R_1, \ldots, R_n \}$ as above, the surrogate approximation is   
$$
\hat{u}_\rho = \sum_{i\in \mathcal{I}_R} \alpha_i(R) u_{R_i},
$$
where 
$$ \alpha_i(R)=  \frac{1 }{\displaystyle \vert R-R_i \vert^{1/4} \Vert \nabla u_{R_i} \Vert_{L^4(B)} \left( \sum_{j\in \mathcal{I}_R} \frac{ 1 }{ \vert R-R_j \vert^{1/4} \Vert \nabla u_{R_j} \Vert_{L^4(B)}} \right)  }.$$

\subsubsection*{Error Estimates}
From above definitions, the  bilinear's form  error is estimated by 
\begin{equation*}
    \sup_{v\in V, \Vert v \Vert_V \leq 1} \vert  \mathcal{A}_R ( \hat{u}_R,v)- \mathcal{L}(v) \vert \leq  \frac{(2\pi)^{1/4}\rho }{\displaystyle \frac{1}{k}  \sum_{j\in \mathcal{I}_R} \frac{1}{\vert R - R_j \vert^{1/4} \Vert \nabla u_{R_j} \Vert_{L^4(B)}} } = \varepsilon_R.
\end{equation*}

\subsubsection*{Numerical Results}
For the numerical test we let $\rho = 6$ and take $m=10$ measurement points $x_i = (\cos(2\pi(i/10)), \sin(2\pi(i/10)))$ for $i \in\{ 0\ldots, 9\}$. We use the FEM implementation of the first example with a triangular mesh of $56,206$ elements and  $111,610$ vertices. The latter implementation generated the synthetic measurement for $R_0=0.725$ and computed $F(u_R)$ for the surrogate coefficients. In addition, we let $\sigma = 0.01$ as the measurement noise standard deviation in \eqref{modelootravez1}.
We assume a uniform apriori distribution supported on $[0,1]$ and implemented MCMC algorithm t-walk to get $800,000$ samples from the posterior distribution. Finally, we considered $k=2$ in the nearest neighbors' computation.

Table~\ref{table:2} compares preprocessing, MCMC exploration, and total times for different numbers $n$ of surrogate elements. The last line reports the sum of all executions for every value of $n$.

\begin{table}[H]
\centering
\caption{\label{table:2}Column $n$ is the number of solutions in the surrogate model, column a) is the time to compute all surrogate elements, $u_{\theta_i}$, in FreeFem, column  b) is the time required for the MCMC exploration with $800,000$ iterations, column  c) is the total time required, and column $\varepsilon_\rho$ is the error estimate. The sum in the last line is the time required to compute all models.}
\begin{tabular}{l c c c c}
\hline
\textbf{n} & \textbf{a)} & \textbf{b)} & \textbf{c)} & \textbf{$\epsilon$}\\
\hline
n = $5$ & 1.034 min  & 2.1 min  & 3.134 min & $(3/2^2)^{1/4}$ \\
n = $9$ & 1.466 min & 3.1 min & 4.466 min & $(3/2^3)^{1/4}$\\
n = $17$ & 2.622 min & 3.3 min & 5.922 min & $(3/2^4)^{1/4}$\\
n = $33$  & 5.038 min & 3.4 min & 8.438 min & $(3/2^5)^{1/4}$ \\
n = $65$  & 10.524 min & 3.4 min & 13.924 min & $(3/2^6)^{1/4}$ \\
n = $129$  & 21.005 min & 3.4 min & 24.405 min & $(3/2^7)^{1/4}$ \\
\hline
 & & Sum = & 60.289 min &
\end{tabular}
\end{table}

As before, the brute force approach of using the exact forward map for $800,000$ iterations would take approximately 71 days. Therefore, the time to sample the posterior distribution with $n=129$ surrogate elements is about $5\cdot 10^3$ times faster, and performing all experiments for different values of $n$ is about $1,800$ times faster.

Figure~\ref{fig:infer_R_2} presents the posterior distributions for different numbers of surrogate elements $n$ and a burn-in of
10,000 samples, after inspection of the MCMC output in each case.  We included a sample of the posterior distribution obtained after $800,000$ iterations with the exact analytical solution \eqref{eq:explicit}  for comparison. 
In Figure~\ref{fig:infer_R_2} (a) to (f), we can see how $\hat{P}_{\Theta\mid Y}$ approaches to the true posterior distribution $P_{\Theta\mid Y}$ as $n$ increases. In contrast with the results in example~\ref{ICV}, the convergence of P1 to P2 is slower. For $n=33$ the surrogate's posterior $\hat{P}_{\Theta\mid Y}$ present significant differences with the analytical model posterior $\hat{P}_{\Theta\mid Y}$.We obtain an acceptable result only for $n=129$. Moreover, from $n=65$ to $n=129$, we observe very little change in the posteriors $\hat{P}_{\Theta\mid Y}$ indicating convergence. 
Because Proposition~\ref{prop:consistency} assures convergence and computation time is small, running for increasing values of $n$ until $\hat{P}_{\Theta\mid Y}$ converges is a sensible strategy to find the true posterior.

\begin{figure}[!htb]
\centering
\begin{subfigure}{0.49\textwidth}
\centering
\includegraphics[width=\textwidth]{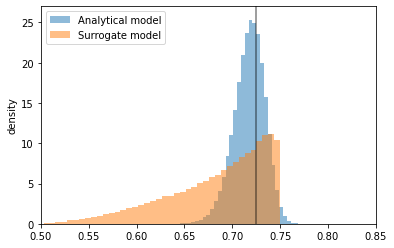}
\caption[]%
{{\small $n=2^2+1$}}    
\label{fig:mean and std of net14-2}
\end{subfigure}
\hfill
\begin{subfigure}{0.49\textwidth}  
\centering 
\includegraphics[width=\textwidth]{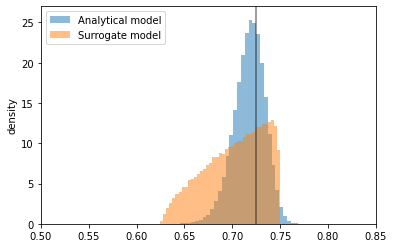}
\caption[]%
{{\small $n=2^3+1$}}    
\label{fig:mean and std of net24-2}
\end{subfigure}

\vskip\baselineskip
\begin{subfigure}{0.49\textwidth}   
\centering 
\includegraphics[width=\textwidth]{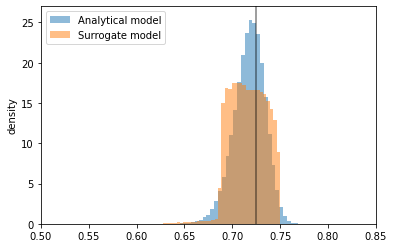}
\caption[]%
{{\small $n=2^4+1$}}    
\end{subfigure}
\hfill
\begin{subfigure}{0.49\textwidth}   
\centering 
\includegraphics[width=\textwidth]{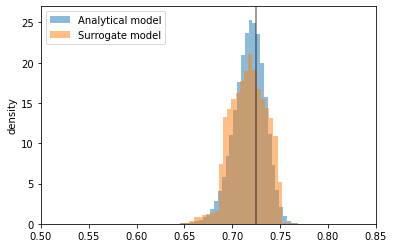}
\caption[]%
{{\small $n=2^5+1$}}    
\label{fig:mean and std of net44-1}
\end{subfigure}

\vskip\baselineskip
\begin{subfigure}{0.49\textwidth}   
\centering 
\includegraphics[width=\textwidth]{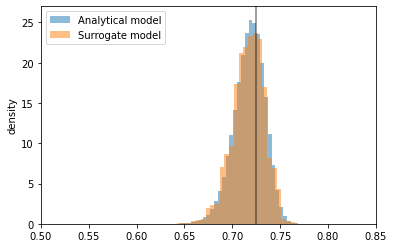}
\caption[]%
{{\small $n=2^6+1$}}    
\end{subfigure}
\hfill
\begin{subfigure}{0.49\textwidth}   
\centering 
\includegraphics[width=\textwidth]{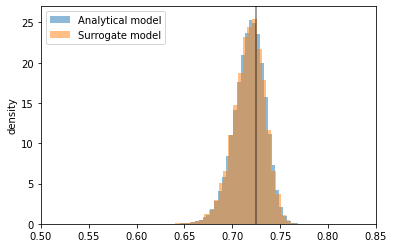}
\caption[]%
{{\small $n=2^7+1$}}    
\label{fig:mean and std of net44-2}
\end{subfigure}
\caption[]
{\small  \label{fig:infer_R}  Histograms of the MCMC samples for the posterior distribution with two designs, increasing the number of elements in our design $\mathcal{D}$.  A binary partition of $[0,1] = \Theta$ was used. Convergence to exact posterior is observed with few model evaluations. The vertical line is the true value of $R$.} 
\label{fig:infer_R_2}
\end{figure}
\FloatBarrier

\subsection{Locating a disk anomaly in EIT}

This section presents an example of Bayesian inference to solve an inverse problem for two unknown parameters. We implement two different designs ${\mathcal D}$ satisfying the $\varepsilon$-approximation property to test further the method's performance. We assess the efficiency by comparing a brute force approach with the numerical FEM implementation. In this case, no exact analytical expression for the forward map is available, so the convergence of the two numerical implementations tests the method's efficacy.

\subsubsection*{Inverse Problem}
The inverse problem is, in a circular domain, inferring the center $c=(c_1,c_2)$ of a small disk anomaly of different conductivity from a series of potential measurements $u(x)$ on $m$ boundary points $\{x_i\}_{i=1}^m\subset \partial B$, see Figure \ref{fig:incl_invprob}. We assume that the radius $r$ of the anomaly is known, and use the notation $\lambda_c$ for the conductivity to stress that the unknown parameter is $c$. 

\begin{figure}[!htp]
    \centering
    \includegraphics[scale = 0.6]{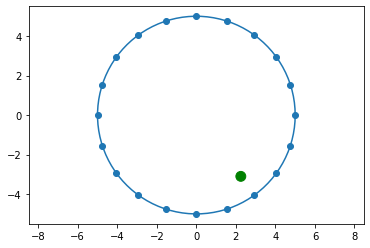}
    \caption{Domain $B$ with $20$ observation points (blue) and a disk anomaly (green)  with center at $c=(2.25,-3.1)$ and radius $r=0.25$.}
    \label{fig:incl_invprob}
\end{figure}

\subsection*{Forward Map} \label{FM3}
The forward map model is the same as the one presented in Section~\ref{EIT:disk} with three differences; first, we assume that $B$ is a disk of radius five. Second, the conductivity has the form 
\begin{equation} \label{cond4}
 \lambda_{c}(x) = \left\{ \begin{array}{cc}
1 + \rho     &  \mbox{ if }  |x-c|\leq r \\
1    &  \mbox{ otherwise} .
\end{array} \right. 
\end{equation}
Third, we let  $f(x)=x_1^4-6x_1^2x_2^2+x_2^4$ with $x =(x_1,x_2) \in \partial B$, furthermore $f$ satisfies the consistent condition \eqref{consistent}.
Hence, the admissible set of parameters is $\Theta=\{c\in{\mathbb R}^2 \text{ s.t. } |c| < 5-r \}$ and the exact FM ${\mathcal F}:\Theta \to V$  maps $c\in\Theta$ to the weak solution of \eqref{problem1} with $\lambda=\lambda_c$ and $f$ as above. 

\subsubsection*{Bilinear's Form Continuity Estimate}  In this case, the bilinear form  is 
$$ \mathcal{A}_c (u,v) = \int_B \lambda_c \nabla u \cdot \nabla v \, dx$$
for every  $u,v \in H_0^1 (B)$, and
the continuity estimate \eqref{modulodecontinuidad}
becomes
\begin{equation*}
    \sup_{v\in V, \Vert v\Vert_V\leq 1 } \vert \mathcal{A}_{c}(u,v)-\mathcal{A}_{c'}(u,v) \vert \leq   \Vert \lambda_{c} - \lambda_{c'} \Vert_{L^4(B)} \Vert \nabla u \Vert_{L^4(B)} \sup_{v\in V, \Vert v\Vert_V\leq 1 }\Vert \nabla v \Vert_{L^2(B)}, 
\end{equation*}    
where to obtain the latter inequality, we argued as in Section~\ref{exmaple:ICS}. Comparing latter inequality with \eqref{modulodecontinuidad}, we get 
\begin{equation}
\label{CGF}
C=1,\quad G(c,c') = \Vert \lambda_{c} - \lambda_{c'} \Vert_{L^4(B)}, \text{ and } F(u)= \Vert \nabla u \Vert_{L^4(B)}.
\end{equation}
There are some alternative expressions to compute the value of  $G(c, c')$. Let $D_r$ and $D_r'$ be two disks of radius $r$ centered at $c,c'\in \Theta$, and denote their symmetric difference by $D_r \triangle D_r' = (D_r\cup D_r')\setminus (D_r\cap D_r') $. Hence, 
$$
 \Vert \lambda_{c} - \lambda_{c'} \Vert_{L^4(B)}
 = \rho\left(\int_B \chi_{D_r \triangle D_r'} dx\right)^\frac{1}{4} = \rho |D_r \triangle D_r'|^\frac{1}{4},
$$
where $\chi_\omega$ is the characteristic function of the set $\omega\subset {\mathbb R}^2$, namely, 
$$
\chi_\omega(x) = \left\{
\begin{array}{ll}
1 & \text{ if } x \in \omega,\\
0 & \text{ if } x \in {\mathbb R}^2\setminus\omega.
\end{array}
\right.
$$

\subsubsection*{Designs}
For given $r< 5$, let $\Theta=\{c\in{\mathbb R}^2 \text{ s.t. } |c| < 5-r \}$ be the admissible parameters set. Because of $G(R,R')=g(d(f(R),f(R')))  =  \Vert \lambda_{c} - \lambda_{c'} \Vert_{L^4(B)}$, we identify $g(x)=x$, $d(a,b)=\Vert a-b\Vert_{L^4(B)}$ and $f(c) =\lambda_c$ in the structure hypothesis \eqref{eq:Gstructure}. Therefore, the metric space is $L=L^4(B)$ equipped with the distance $d$. The set
$f(\Theta)$ is totally bounded since it is the image of a bounded set in a finite-dimensional space. Therefore, for any $\epsilon>0$, there exists a $\epsilon$-net, denoted by $P_\epsilon$, and  we let ${\mathcal D}= f^{-1}(P_\epsilon)$ be the design satisfying the ($\epsilon$-)approximation property (cf. Definition~\ref{Def:ApproxProperty}). 

Since $B$ is a bounded disk, we construct $P_\epsilon$ by giving a set of points ${\mathcal C}=\{c^k\}_{k=1}^n\subset \Theta$ such that 
$$
\min_{1\leq k\neq j\leq n}\{|c_j-c_k|\}\le  r_\epsilon,
$$ 
where, if we denote by $D_j$ and $D_k$ two intersecting disks of radius $r$ centered at $c_j$ and $c_k$, respectively, $r_\epsilon$ is defined by the condition 
$$
|D_j \triangle D_k| = 
d\sqrt{r^2-r_\epsilon^2}+2r_\epsilon^2\left[\pi-\cos^{-1}\left(\frac{r_\epsilon}{2 r}\right)\right] < \frac{\epsilon^4}{\rho^4}.
$$
The number of elements $n$ in ${\mathcal C}$ is also defined by the latter condition. We present two different designs ${\mathcal D}_1$ and ${\mathcal D}_2$ for the numerical tests. 

\subsubsection*{Surrogate Model}

Given the design ${\mathcal D} = \{ c_1, \ldots,c_n \}$ as above, the surrogate approximation is   
$$
\hat{u}_c = \sum_{i\in \mathcal{I}_c} \alpha_i(c) u_{c_i},
$$
where 
$$ \alpha_i(c)=  \frac{1 }{\displaystyle \Vert \lambda_{c} - \lambda_{c_i} \Vert_{L^4(B)} \Vert \nabla u_{c_i} \Vert_{L^4(B)} \left( \sum_{j\in \mathcal{I}_c} \frac{ 1 }{ \Vert \lambda_{c} - \lambda_{c_j} \Vert_{L^4(B)} \Vert \nabla u_{c_j} \Vert_{L^4(B)}} \right)  }.$$

\subsubsection*{Error Estimates}
From above definitions, the bilinear's form  error is estimated by 
\begin{equation*}
    \sup_{v\in V, \Vert v \Vert_V \leq 1} \vert  \mathcal{A}_c ( \hat{u}_c,v)- \mathcal{L}(v) \vert \leq  \frac{1 }{\displaystyle \frac{1}{k}  \sum_{j\in \mathcal{I}_c} \frac{1}{\Vert \lambda_{c} - \lambda_{c_j} \Vert_{L^4(B)}\Vert \nabla u_{c_j} \Vert_{L^4(B)}} } = \varepsilon_c.
\end{equation*}

\subsubsection*{Results with Design ${\mathcal D}_1$} In this design we consider a set of centers ${\mathcal C}$ generated by a triangulation algorithm. Figure~\ref{figure7} presents the designs for different number $n$ of elements. 
\begin{figure}[!htb] 
\centering
\begin{subfigure}{0.4\textwidth}
\centering
\includegraphics[width=\textwidth]{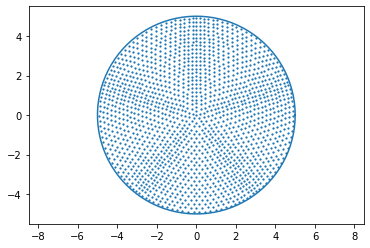}
\caption[]%
{{\small $n=1,626$}}    
\end{subfigure}
\hfill
\begin{subfigure}{0.4\textwidth}  
\centering 
\includegraphics[width=\textwidth]{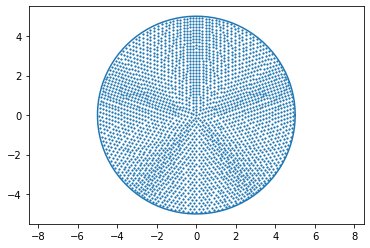}
\caption[]%
{{\small $n=2,641$}}    
\end{subfigure}
\vskip\baselineskip
\begin{subfigure}{0.4\textwidth}   
\centering 
\includegraphics[width=\textwidth]{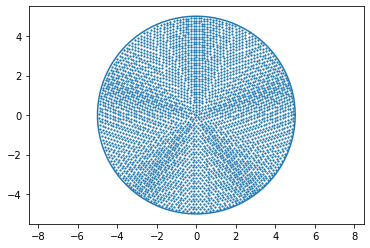}
\caption[]%
{{\small $n=3,516$}}    
\end{subfigure}
\hfill
\begin{subfigure}{0.4\textwidth}   
\centering 
\includegraphics[width=\textwidth]{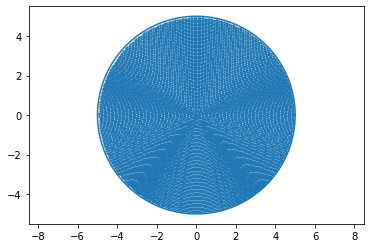}
\caption[]%
{{\small $n=12,426$}}    
\end{subfigure}
\caption[]
{\small Center points in the set ${\mathcal C}$ for design $\mathcal{D}_1$ and different values of $n$.} 
\label{figure7}
\end{figure}
For the numerical test we let $\rho = 6$ and take $m=20$ measurement points $x_i = (\cos(2\pi(i/20)), \sin(2\pi(i/20)))$ for $i \in\{ 0\ldots, 9\}$. We use the FEM implementation as before with a triangular mesh of $17,300$ elements and  $89,450$ vertices. The latter implementation generated the synthetic measurement for $c_0=(2.5,-3.1)$ with $r=0.25$ (cf. Figure~\ref{fig:incl_invprob}) and computed $F(u_R)$ for the surrogate coefficients. In addition, we let $\sigma = 0.01$ as the measurement noise standard deviation in \eqref{modelootravez1}.
For the Bayesian model, we assumed an uniform apriori supported on $\Theta$ and implemented MCMC algorithm t-walk to get $400,000$ samples from the posterior distribution. Finally, we considered $k=3$ in the nearest neighbors' computation.

\begin{table} [!htb]
\centering
\caption{Column $n$ is the number of solutions in the surrogate model, column a) is the time to compute all surrogate elements, $u_{c_i}$, in FreeFem, column  b) is the time required for the MCMC exploration with $400,000$ iterations, column  c) is the total time required, and column $\varepsilon_c$ is the error estimate. The sum in the last line is the time required to compute all models.}
\begin{tabular}{c c c c c} 
 \hline
 \textbf{n} & \textbf{a)} & \textbf{b)} & \textbf{c)} & \textbf{$\varepsilon_c$} \\  
 \hline
$n = 1626$ & 1.993 hr  & 0.120 hr  & 2.113 hr & 6.113 \\
$n = 2641$ & 3.363 hr & 0.2036 hr & 3.566 hr  & 5.786 \\
$n = 3516$ & 3.904 hr & 0.243 hr & 4.147 hr  & 5.591 \\
 $n = 12426$ & 13.403 hr & 0.838 hr &  14.241 hr  & 4.791 \\
 \hline
 & & Sum = & 24.067 hr &
\end{tabular}
\label{table:4}
\end{table}
Table~\ref{table:4} compares preprocessing, MCMC exploration, and total times for different numbers $n$ of surrogate elements. The last line reports the sum of all executions for every value of $n$.

The brute force approach of using the exact forward map for $400,000$ iterations would take approximately 24 days. Therefore, the time to sample the posterior distribution with $n=3,516$ surrogate elements is about $129$ times faster. Even performing all experiments for different values of $n$ is 24 times faster than using the exact forward map once. 
\begin{figure}[!b]
  \centering

  \begin{subfigure}[t]{0.45\textwidth}
    \centering\includegraphics[width=\textwidth]{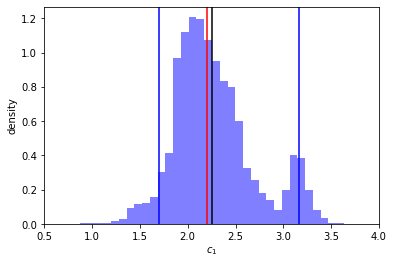}
    \caption{$c_1$ for $n=1626$}
  \end{subfigure}\hfill
  \begin{subfigure}[t]{0.45\textwidth}
    \centering\includegraphics[width=\textwidth]{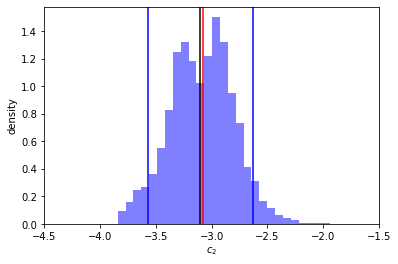}
    \caption{$c_2$ for $n=1626$}
  \end{subfigure}

  
  
  \begin{subfigure}[t]{0.45\textwidth}
    \centering\includegraphics[width=\textwidth]{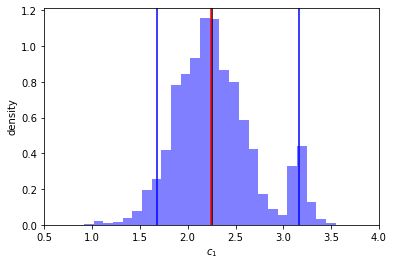}
    \caption{$c_1$ for $n=2641$}
  \end{subfigure}\hfill
  \begin{subfigure}[t]{0.45\textwidth}
    \centering\includegraphics[width=\textwidth]{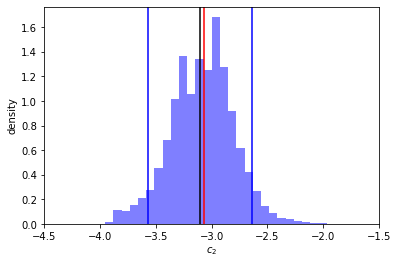}
    \caption{$c_2$ for $n=2641$}
  \end{subfigure}



  \begin{subfigure}[t]{0.45\textwidth}
    \centering\includegraphics[width=\textwidth]{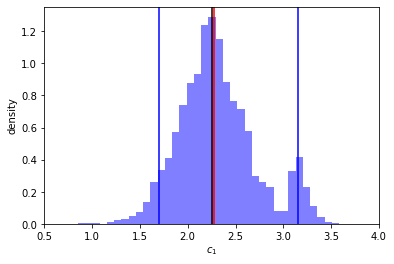}
    \caption{$c_1$ for $ n= 3516 $}
  \end{subfigure}\hfill
  \begin{subfigure}[t]{0.45\textwidth}
    \centering\includegraphics[width=\textwidth]{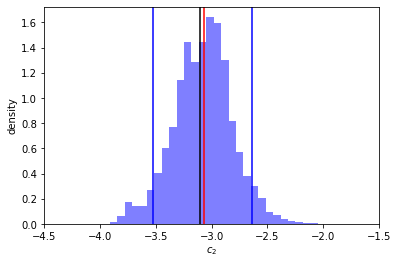}
    \caption{$c_2$ for $ n= 3516 $}
  \end{subfigure}
	\end{figure}


\begin{figure}[ht] \ContinuedFloat 
\centering
   \medskip
  \begin{subfigure}[t]{0.45\textwidth}
    \centering\includegraphics[width=\textwidth]{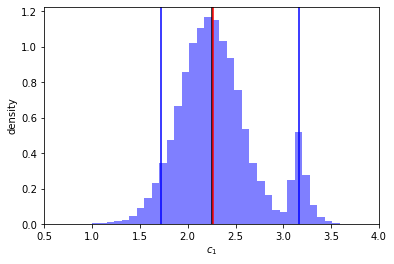}
    \caption{$c_1$ for $ n= 12426 $}
  \end{subfigure}\hfill
  \begin{subfigure}[t]{0.45\textwidth}
    \centering\includegraphics[width=\textwidth]{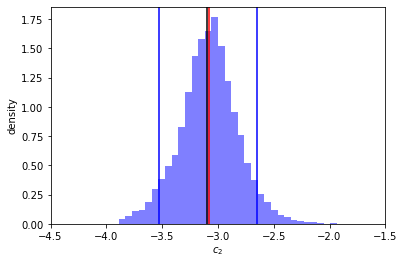}
    \caption{$c_2$ for $ n= 12426 $}
  \end{subfigure}
\caption{ Histograms of the MCMC samples for the posterior distribution with three designs, increasing the number of elements in our design $\mathcal{D}$. The black vertical line is the true value of $h$, the red line is the median and the blue lines are the 0.05 and 0.95 quantiles.}
\label{figure8}
\end{figure}

Figure~\ref{figure8} presents the posterior distributions for different numbers of surrogate elements $n$ and a burn-in of
10,000 samples. The posterior distributions do not change significantly after $n=2641$; the posterior distributions only get more sharply defined for larger values of $n$. Therefore, we have the convergence to the actual posterior distribution for relatively small values of $n$. The posterior distributions have relatively small support; interestingly, the $c_1$ marginal is bimodal, leading to a phantom in the reconstruction. 

\subsubsection*{Results with Design ${\mathcal D}_2$} In this design, we consider a set of centers ${\mathcal C}_2$ generated by a square grid. Figure~\ref{figure9} presents the designs for different numbers $n$ of elements. If the distance between centers is less than $r_\epsilon$ above, the design will satisfy the ($\epsilon$)-approximation property. Since $f$ is not the identity, the designs ${\mathcal D}_1$ and ${\mathcal D}_2$ in the metric space $L$ are not equivalent. 
\begin{figure}[!ht] 
\centering
\begin{subfigure}{0.4\textwidth}
\centering
\includegraphics[width=\textwidth]{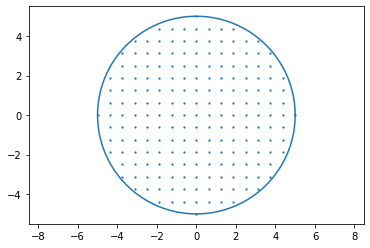}
\caption[]%
{{\small $n=197$}}    
\end{subfigure}
\hfill
\begin{subfigure}{0.4\textwidth}  
\centering 
\includegraphics[width=\textwidth]{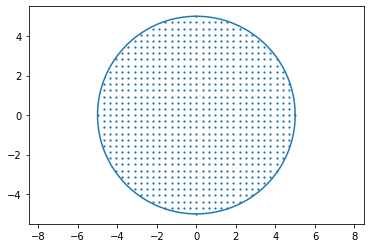}
\caption[]%
{{\small $n=797$}}    
\end{subfigure}
\vskip\baselineskip
\begin{subfigure}{0.4\textwidth}   
\centering 
\includegraphics[width=\textwidth]{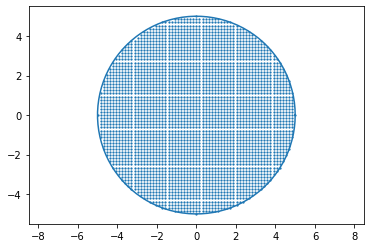}
\caption[]%
{{\small $n=3209$}}    
\end{subfigure}
\hfill
\begin{subfigure}{0.4\textwidth}   
\centering 
\includegraphics[width=\textwidth]{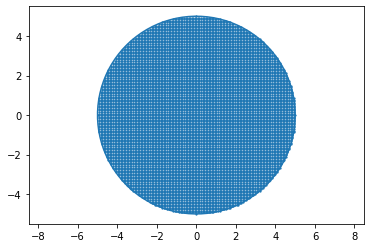}
\caption[]%
{{\small $n=12853$}}    
\end{subfigure}
\caption[]
{\small Center points in the set ${\mathcal C}$ for design $\mathcal{D}_2$ and different values of $n$.} 
\label{figure9}
\end{figure}
In this experiment, we used the numerical setting for design ${\mathcal D}_1$ and the same observations for the inference problem. 

\begin{table} [!htb]
\centering
\caption{Column $n$ is the number of solutions in the surrogate model, column a) is the time to compute all surrogate elements, $u_{c_i}$, in FreeFem, column  b) is the time required for the MCMC exploration with $400,000$ iterations, column  c) is the total time required, and column $\varepsilon_c$ is the error estimate. The sum in the last line is the time required to compute all models.}
\begin{tabular}{c c c c c} 
 \hline
 \textbf{n} & \textbf{a)} & \textbf{b)} & \textbf{c)} & \textbf{$\varepsilon_\rho$} \\  
 \hline
$n = 197$ & 0.358 hrs  & 0.033 hrs  & 0.391 hrs & 6.728 \\
$n = 797$ & 0.9120 hrs & 0.092 hrs & 1.004 hrs  & 6.648\\
$n = 3209$ & 3.394 hrs & 0.3072 hrs & 3.701 hrs  & 5.778 \\
$n = 12853$ & 16.658 hrs & 1.020 hrs & 17.689 hrs  & 4.890 \\
 \hline
 & & Sum = & 22.785 hrs &
\end{tabular}
\label{table:5}
\end{table}
Table~\ref{table:5} compares pre-processing, MCMC exploration, and total times for different numbers $n$ of surrogate elements. The last line reports the sum of all executions for every value of $n$.
In this case, the performance is similar to the one of  ${\mathcal D}_2$. The execution times are of the same order of magnitude, and the brute force approach would take approximately 24 days as before. Therefore, the time to sample the posterior distribution with $n=3,209$ surrogate elements is about $155$ times faster. 
\begin{figure}[!b]
  \centering
  \begin{subfigure}[t]{0.5\textwidth}
    \centering\includegraphics[width=\textwidth]{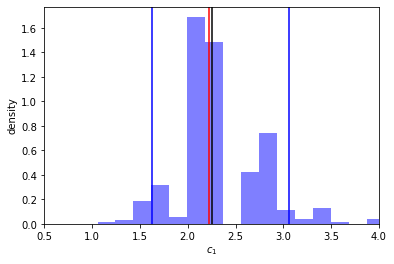}
    \caption{$c_1$ for $n=197$}
  \end{subfigure}\hfill
  \begin{subfigure}[t]{0.5\textwidth}
    \centering\includegraphics[width=\textwidth]{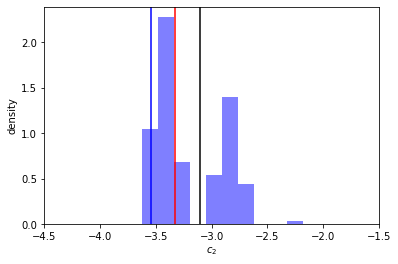}
    \caption{$c_2$ for $n=197$}
  \end{subfigure}

  
  
  \begin{subfigure}[t]{0.5\textwidth}
    \centering\includegraphics[width=\textwidth]{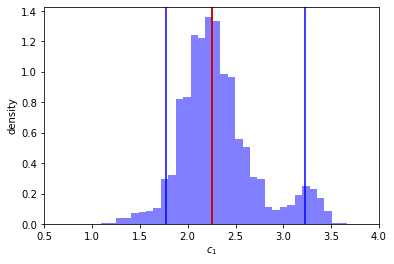}
    \caption{$c_1$ for $n=797$}
  \end{subfigure}\hfill
  \begin{subfigure}[t]{0.5\textwidth}
    \centering\includegraphics[width=\textwidth]{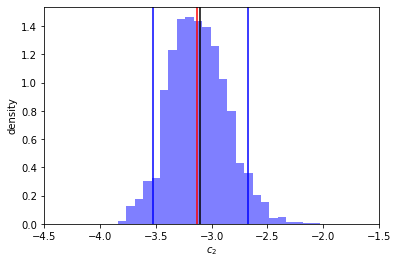}
    \caption{$c_2$ for $n=797$}
  \end{subfigure}

  \begin{subfigure}[t]{0.5\textwidth}
    \centering\includegraphics[width=\textwidth]{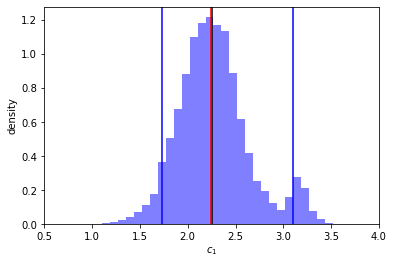}
    \caption{$c_1$ for $ n= 3209 $}
  \end{subfigure}\hfill
  \begin{subfigure}[t]{0.5\textwidth}
    \centering\includegraphics[width=\textwidth]{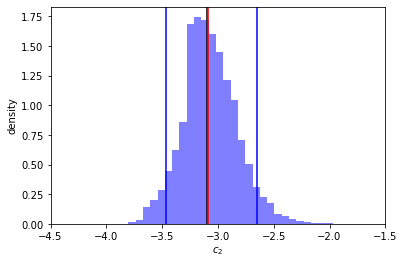}
    \caption{$c_2$ for $ n= 3209 $}
  \end{subfigure}
	\end{figure}


\begin{figure}[ht] \ContinuedFloat 
\centering
   \medskip
  \begin{subfigure}[t]{0.5\textwidth}
    \centering\includegraphics[width=\textwidth]{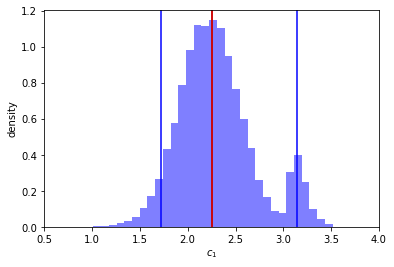}
    \caption{$c_1$ for  $n= 12853 $}
  \end{subfigure}\hfill
  \begin{subfigure}[t]{0.5\textwidth}
    \centering\includegraphics[width=\textwidth]{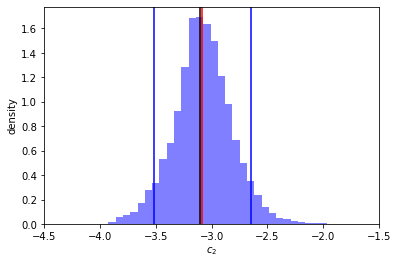}
    \caption{$c_2$ for  $n= 12853 $}
  \end{subfigure} 
\caption{ Histograms of the MCMC samples for the posterior distribution with three designs, increasing the number of elements in our design $\mathcal{D}$. The black vertical line is the true value of $h$, the red line is the median and the blue lines are the 0.05 and 0.95 quantiles.}
\label{figure9}
\end{figure}
Figure~\ref{figure9} presents the posterior distributions for different numbers of surrogate elements $n$ and a burn-in of
10,000 samples. The posterior distributions do not change significantly after $n=3,209$; the distributions only get more sharply defined and coincide with the posterior distributions presented in Figure~\ref{figure8} as expected by  Proposition~\ref{prop:consistency}.
As in the previous experiment, the $c_1$ marginal is bimodal, leading to the same phantom in the reconstruction.

\subsubsection*{Reconstruction and efficiency}
Figure~\ref{fig:adfig} shows the reconstruction with uncertainty quantification for design ${\mathcal D}_1$ with $n=12426$ elements. Transparent green disks correspond to the last $500$ samples in the MCMC considering an IAT$=49$, the red circumference represents the true inclusion, and the phantom in the reconstruction is clearly visible.   
\begin{figure}[htb]
    \centering
    \includegraphics[scale = 0.6]{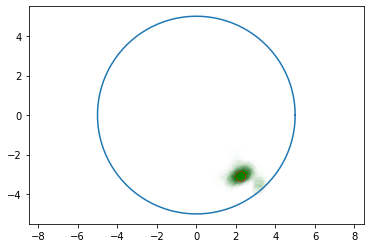}
    \caption{Reconstruction with uncertainty quantification of the disk inclusion.}
    \label{fig:adfig}
\end{figure}

Table~\ref{table:4} and  Table~\ref{table:5} show that the most significant increase in computational time is when computing the surrogate elements $u_c$ in the pre-processing step. Therefore, by doubling the length of the MCMC chain for $n=3,516$ in Table~\ref{table:4}, there is only an increase of 0.243 hr, which amounts to a total time of 4.39 hr. Meanwhile, the time using the exact forward map would double to approximately 48 days, about 262 times faster. Therefore, for longer the MCMC chain, the surrogate's performance efficiency increases sharply.

\section{Discussion} \label{sec:discussion}
The paper introduces the Physics-Based Surrogate Model (PBSM), a computationally efficient method for approximating forward maps associated with linear elliptic partial differential equations. Designed for use in Bayesian inferences for inverse problems with uncertainty quantification, the PBSM aims to mitigate the computational expense of exact forward map evaluations. However, they are only required at a preprocessing step that can be parallelized. Hence, the method's efficiency, especially for extended Markov Chain Monte Carlo (MCMC) explorations.

Key aspects of the PBSM include its requirement of a linear elliptic PDE, a continuity estimate of the equation's bilinear weak form on the unknown parameters, and a design.  Unlike other data-fit surrogate models, the method derives surrogate's coefficients analytically, maintaining the physics of the problem in the weak form of the equations. The adapted notion of distance for the problem, represented by the function $G$ in equation \eqref{modulodecontinuidad}, plays a crucial role in the proposed method. It serves a dual purpose in the approach: firstly, in the construction of a design that satisfies the approximation property, and secondly, in defining the $\epsilon$-approximation property. As a result, the method is informed and grounded in the physics of the problem through two essential aspects: the coefficients and the design.

While the examples focus on one or two parameters, the theory extends to parameter sets of arbitrary dimensions. The case of infinite dimension is also covered, but in practice, constructing a design represents a big challenge. The PBSM approximations are not continuous with respect to parameters, emphasizing that continuity is unnecessary for good performance in the inference process. Only a global control on the forward map's error, such as the one in the weak form of the forward problem, is enough.

In contrast to other surrogate model construction methods, such as polynomial chaos expansions, kriging, support vector machines, and (deep) neural networks, our proposed method relies on an entirely analytical approach for constructing surrogate coefficients and establishing error estimates. The results, including the existence of the posterior distribution, its consistency, and error estimates, provide assurance regarding the method's performance. Notably, the PBSM distinguishes itself from typical multi-fidelity methods by enabling accuracy and convergence tests without exact forward map evaluations. While convergence can be confirmed through theoretical results by testing increasingly refined designs, the method's exceptional speed makes this approach practical. The presented examples, particularly the second one, demonstrate that exploring various designs for convergence is three to four orders of magnitude faster than a single brute force test.

In contrast to methods that often use generic approaches like Latin squares or low-discrepancy sequences for design construction, the PBSM requires designs that satisfy the approximation property. While constructing such designs can be challenging, the proposed method is effective, especially for low-dimensional examples. However, high-dimensional problems introduce two potential challenges. First, design construction may become computationally expensive, requiring numerous elements to meet the approximation property, and it can strain computer memory and processing speed. The second challenge pertains to the nearest-neighbor algorithm, as no exact algorithm with acceptable performance in high dimensions is currently known. 


The present research suggests four natural extensions:
\emph{Extension to Vector-Valued PDEs:} The theory and examples developed in this work focus on scalar linear PDE problems. However, the approach can be extended with minimal modifications to vector-valued PDEs, such as those modeling linear elasticity and electrostatics. This extension can facilitate addressing more realistic examples. \emph{High-Dimensional Inverse Problems:} Testing the method for high-dimensional inverse problems is a logical next step, presenting challenges in numerical methods, particularly in design construction and nearest-neighbor computation. \emph{Exploration of Posterior and Design Improvement:} The approximation property is primarily required near the support of the posterior distribution. Therefore, the global approximation property hypothesis can be seen as somewhat restrictive. Exploring the posterior and refining the design while maintaining consistency and error estimate results could potentially relax this hypothesis and enhance the method's computational efficiency. \emph{Extension to Nonlinear Operators:} Extending the PBSM method to handle nonlinear operators of elliptic type or time-dependent problems is an intriguing avenue for future research. However, the lack of linearity will requiere e a different coefficient definition and a new error analysis. These extensions can further broaden the applicability and efficiency of the PBSM method, addressing more complex and diverse problems in various scientific and engineering domains.

Efficient and mathematically sound approximations in inverse problems with Uncertainty Quantification are achieved by focusing on the analytical and structural features of the forward map. Leveraging these features enables the derivation of fast surrogate approximations, accompanied by an error and consistency theory.  These surrogates retain sufficient information about the solution's landscape, allowing for the recovery of posterior distribution approximations with minimal loss of accuracy while accelerating computations by several orders of magnitude.

\section*{Acknowledgments}
We thank Robert Scheichl for his valuable suggestions that have enhanced this manuscript. A. Galaviz acknowledges the fellowship (CVU 638587) from CONAHCYT. Both A. Galaviz and A. Capella received partial support from CONAHCyT, México, grant CF-2023-G-122.

\bibliographystyle{unsrt}  
\bibliography{references}

\end{document}